\documentclass[reqno,11pt]{amsart}

\usepackage{amsmath}
\usepackage{amssymb}
\usepackage{amsthm}

\usepackage{ifthen}
\usepackage{graphicx}

\usepackage[utf8]{inputenc}
\usepackage[T1]{fontenc}

\usepackage{psfrag}
\numberwithin{equation}{section}

\newtheorem{lm}{Lemma}
\newtheorem{thm}{Theorem}[section]
\newtheorem*{thm*}{Theorem}

\newtheorem*{corollaire*}{Corollary}

\newtheorem{prop}[thm]{Proposition}
\newtheorem*{prop*}{Proposition}

\begin{document}

\title[]{
Coherence transition in degenerate diffusion equations with mean field coupling}
\author{Khashayar Pakdaman}
\address{Institut Jacques Monod, Universit\'e Paris 7--Denis Diderot,
Bat. Buffon, 15 rue H\'el\'ene Brion - 75013 Paris, France}
\author{Xavier Pellegrin}
\address{Institut Jacques Monod, Universit\'e Paris 7--Denis Diderot,
Bat. Buffon, 15 rue H\'el\'ene Brion - 75013 Paris, France}
\date{February 28th, 2012}
\begin{abstract}
We introduce non-linear diffusion in a classical diffusion advection model with non local aggregative coupling on the circle, 
that exhibits a transition from an uncoherent state to a coherent one when the coupling strength is increased. 
We show first that all solutions of the equation converge to the set of equilibria, second that the set of equilibria 
undergoes a bifurcation representing the transition to coherence when the coupling strength is increased. These two properties are similar 
to the situation with linear diffusion. 
Nevertheless nonlinear diffusion alters the transition scenari, which are different when the diffusion is sub-quadratic and 
when the diffusion is super-quadratic. 
When the diffusion is super-quadratic, it results in a multistability region that preceeds the pitchfork bifurcation at which 
the uncoherent equilibrium looses stability. 
When the diffusion is quadratic the pitchfork bifurcation at the onset of coherence is infinitely degenerate
 and a disk of equilibria exist for the critical value of the coupling strength.
Another impact of nonlinear diffusion is that coherent equilibria become localized when advection is strong enough, 
a phenomenon that is preculded when the diffusion is linear. 
\end{abstract}

\maketitle

\section{Introduction}
\label{sec_intro}

We consider the equation
\begin{equation}\label{eq:main}
\left \{
\begin{array}{rll}
\partial_t u &= \partial_\theta^2 (u^m) + \partial_\theta \left( u J \ast u \right), \;\; &  t> 0, \, \theta \in [0,2\pi], \\
u(0,\theta) &= u_0(\theta) \;\; & \theta \in [0,2\pi],  \\
u(t, 0) &= u(t, 2 \pi) \;\; & t\geq 0, \\
\partial_\theta u(t, 0) \;\; &= \partial_\theta u(t, 2 \pi) &  t \geq 0, 
\end{array} \right .
\end{equation}
where $m > 0$ and $J \ast u (t,\theta) = K \int_0^{2\pi} \sin(\theta - \varphi) u(t,\varphi) d \varphi, $
with $u_0$ even, $u_0(\cdot) \geq 0$ and $\int_{-\pi}^\pi u_0 (\theta)  d \theta =1$
and $K \geq 0$ is a constant. 
In the following we denote $x_1= x_1(u) = \int_{0}^{2\pi} u(\theta) \cos(\theta) d\theta$ for any $u \in L^1([0,2\pi])$.
The right hand side of this equation is comprised of two components,
with the first representing a nonlinear diffusion and the second one a
nonlocal advection term. The dynamics of the solutions result from
competition between these two terms, as the first one tends to spread
solutions whereas the second one, on the contrary tends to concentrate
them. In this work we examine the changes in the dynamics of the above
equation depending on the parameters $m$ and $K$, which control,
respectively the diffusion and the strength of advection.

\medskip

\par For the special case of linear diffusion, i.e. $m=1$, this
equation arises in various contexts,
including as a mean-field spin X-Y models \cite{silver:468}, a
Doi-Onsager or Smoluchowski model for nematic polymers
\cite{MR2164412,MR2109485,MR2425333},
or a Kuramoto or Sakaguchi model of synchronization
\cite{RevModPhys.77.137,MR1783382,GPP}. Its equilibria set and dynamics have been
analyzed in great detail 
\cite{LuoZZ_2004,MR2109485,silver:468,MR2594897,GPP,MR1115806,Crawford19991,RevModPhys.77.137}
\cite{MR2164412,MR2594897,MR2425333,GPP}.
 The picture that emerges from these
studies is that there are two distinct regimes depending on the value
of $K$. On the one hand, for $K \le \frac 12$, all solutions tend to the
trivial equilibrium $u = 1/(2 \pi)$. On the other hand, for $K> \frac 12$, the
equation admits a unique non trivial equilibrium (up to a rotation)
that attracts all solutions except those that lie on the stable
manifold of the trivial equilibrium. 
The non trivial equilibria are called {\em coherent} because they have a single maximum on $[-\pi , \pi]$, 
corresponding to a region of maximal density, where the population described by $u$ aggregates. 
In this sense, the bifurcation
taking place at $K=\frac 12$ separates a regime where diffusion dominates from
the one where advection promotes coherence. 
This transition whereby
the number of equilibria of the system changes as $K$ crosses a
critical value is the key property of the model. This sudden change of
dynamics accounts for phenomena such as onset of synchrony in coupled
oscillators, and the transition from isotropic to nematic phase in Doi-Onsager models.
We refer to it as a coherence transition.
Its characterization is one of the main motivations for the large
number of studies devoted to this model. In the present work, our
purpose is to analyze the impact of nonlinear diffusion (i.e. $m>1$)
on this transition.

\medskip

\par Nonlinear diffusion equations, with or without advection, have been studied as mathematical models for many important phenomena, 
such as diffusions in porous media or spread of biological populations, 
that show original and complex dynamics 
\cite{MR1417048,MR0682594,MR0455812,MR903393,MR572962,MR0509720,MR1041895,MR783581,MR678053,MR1876751,springerlink:10.1007/s00028-006-0298-z}.
In the context of biological aggregation, equations such as (\ref{eq:main}), with space variable $\theta \in \mathbb R^d$, $m>1$ and various 
coupling functions, have been subject to an intense study, including mathematical results on derivation and
 well-posedness of the equation (\cite{MR700524,2009arXiv0902.2017L,MR2166611,MR2208049,MR2372494,MR2117406,Burger2007939}), 
and several phenomena that were known for the Porous Medium Equation (PME)  \cite{MR0247303,MR0265774,MR0255996,MR2286292} such as 
finite propagation speed of interfaces and existence of travelling waves or stationary solutions 
 \cite{2011arXiv1103.5365B,springerlink:10.1007/s11538-006-9088-6,MR700524,MR1698215,MBMDF}.
Degenerate diffusion ($m>1$) has been introduced as a biologically realistic mechanism for the apparition of clumps in aggregation models, 
i.e. the apparition of small groups of individuals with sharp edges, which had only been previously observed in 
one dimensional models with highly specific coupling kernels.
In case $m=2$ and for a general class of coupling functions $J$, 
Burger and al. \cite{2011arXiv1103.5365B} have studied the existence or non existence
of equilibria in any dimension ($\theta \in \mathbb R^d$), and uniqueness, monotony and compact
 support properties of equilibria in dimension one ($\theta \in \mathbb R$).
Taking $m=3$ and $J(\theta) =K e^{- \vert \theta \vert}$, and considering equation (\ref{eq:main}) with $\theta \in \mathbb R$, 
Topaz and al. \cite{springerlink:10.1007/s11538-006-9088-6} have shown in particular that depending on the coupling strength, 
various equilibria or periodic solutions with compact support exist. 
Degenerate diffusion has been introduced and considered in several other diffusion advection reaction equations. 
See \cite{MR2263432,2009arXiv0902.1878S,MR2501355} for the classical Keller-Segel model of chemotaxis with degenerate diffusion for example.   
Despite this large number of studies, to our knowledge, 
the effect of nonlinear diffusion $m>1$ in equations of the form (\ref{eq:main})
 has never been investigated so far. This is the aim for our paper. 

\medskip

Our main results are summarized in the following theorem:

\begin{thm}\label{thm:main}
\par For all $m \geq 1$ and $K >0$, equation (\ref{eq:main}) is the gradient flow of a free energy and 
all smooth solutions converge to the set of equilibria of (\ref{eq:main}).  
Denoting $\tilde K = K \frac{1}{m} \left( \frac{1}{2\pi} \right)^{-(m-1)}$, 
except for $m=2$ and $\tilde K=2$ this set of equilibria is finite. 
\par For any $m \geq 1$  we have: 
\begin{itemize}
\item For all $\tilde K < 2$, the uncoherent equilibrium $\frac1{2\pi}$ is locally stable.
\item For all $\tilde K > 2$, there is a unique pair of (locally stable) coherent equilibria of (\ref{eq:main}) 
whose bassins of attraction contain an open and dense set of initial conditions. 
\end{itemize}
\par When $1\leq m \leq 2$, for all $\tilde K < 2$ the uncoherent equilibrium is globally stable. 
\par When $m>2$, there is a $0< K_c(m) < 2 $ such that 
\begin{itemize}
\item when $\tilde K < K_c$ the uncoherent equilibrium is globally stable,
\item when $K_c < \tilde K < 2$ the dynamics of (\ref{eq:main}) is bistable: the uncoherent equilibrium $\frac1{2\pi}$ is locally stable, there is 
unique pair of (locally stable) coherent equilibria of (\ref{eq:main}), and all solutions of (\ref{eq:main}) converge to one of those.  
\end{itemize}
\par When $m=1$, the coherent equilibria of (\ref{eq:main}) are positive on $[-\pi,\pi]$ for all $K >0$. 
However, for all $m>1$, there is a $K_l(m) \geq K_c(m)$ such that for all $\tilde K > K_l$, the coherent equilibria of (\ref{eq:main}) are localized in $[- \pi , \pi]$.
\end{thm}

This theorem shows that nonlinear diffusion, besides localizing
equilibria,  modifies the scenario of transition to coherence. This
modification is captured by the two-parameter bifurcation diagrams of
equation (\ref{eq:main})  in the plane $(\tilde K, m)$ represented in fig.
\ref{fig:bif_UBC}
 (wherein the left panel is a magnification of a section of the right
panel). These diagrams show that $(\tilde K,m)=(2,2)$ where a highly
degenerate pitchfork bifurcation takes place is an organizing center
of the dynamics. Three qualitatively distinct regimes labelled as
$(U)$, $(B)$  and $(C)$ come to meet at this point. These regimes
represent, respectively, (i) the parameter range for which the
uncoherent equilibrium $\frac{1}{2\pi}$ is globally asymptotically
stable, (ii) the multistablity regime  where $\frac{1}{2\pi}$ is
locally stable and coexists with two pairs of coherent equilibria,
with one pair being unstable and the other stable, and finally (iii)
the regime of coherence where $\frac{1}{2\pi}$ is unstable and most
solutions converge to either one of a pair of symmetrical stable
coherent equilibria.



\begin{figure}[h!tp]
\begin{center}
\psfrag{K(tilde)}[B][B][1][0]{ {\small $\tilde K$}}
\begin{center}
\includegraphics[scale=1.0]{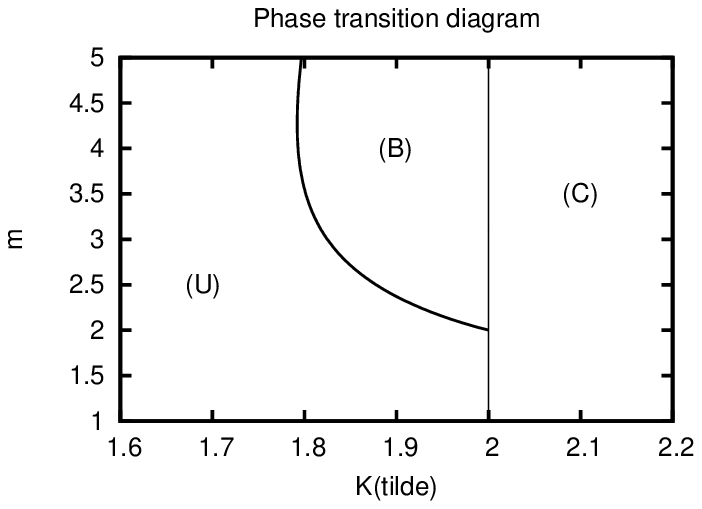}
\includegraphics[scale=1.0]{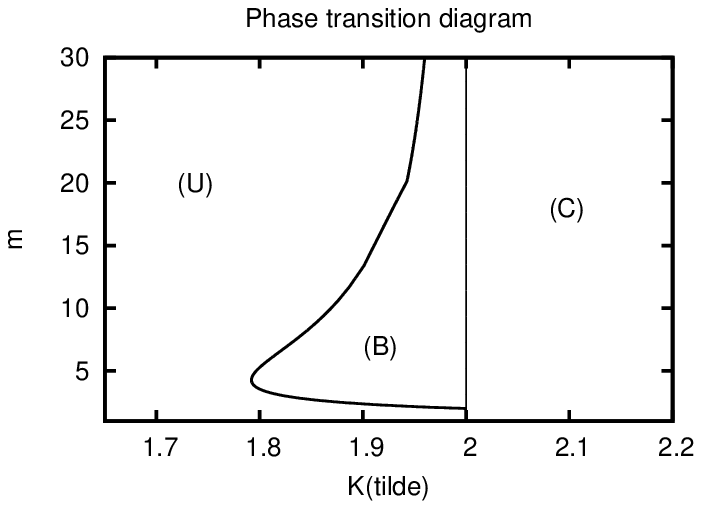}
\end{center}
\end{center}
\begin{center}
\caption{Equilibria bifurcation diagram.
\small{
Solutions of (\ref{eq:main}) converge to the unique (uncoherent)
equilibrium $\frac1{2\pi}$ in region $(U)$.
In region $(C)$, solutions typically converge to the unique coherent
equilibrium, with support strictly included in $[0, 2 \pi]$ in region
$(C_l)$.
In region $(B)$ one coherent and one uncoherent (locally) stable
equilibrium coexist.
$m$ and $K$ are the parameters of (\ref{eq:main}), and $ \tilde K = K
\frac{1}{m} \left( \frac{1}{2\pi} \right)^{-(m-1)}$.
}
}\label{fig:bif_UBC}
\end{center}
\end{figure}

\medskip
\par This paper is organized as follows. In section
\ref{sec:pos_mass_csv}, we show that, as in the case $m=1$ or in the
PME case, equation (\ref{eq:main}) preserves positivity and $L^1$
mass. In section \ref{sec:grad_flow} we introduce the free energy of
(\ref{eq:main}), which allows to write equation (\ref{eq:main}) as a
gradient flow in the formalism of \cite{MR2192294,MR1842429}. A
consequence is that all solutions of (\ref{eq:main}) converge to its
equilibria set when $t \rightarrow + \infty$. We further show that at
each equilibrium point the linearization of (\ref{eq:main}) is
symmetric for an appropriately chosen scalar product. In section
\ref{sec_equilibria}, we give existence and uniqueness results for
equilibria of (\ref{eq:main}) and give analytical formula. This shows
that the pitchfork bifurcation at $\frac1{2\pi}$ is supercritical when
$1<m<2$ (section \ref{sec:bif_m<2}) and subcritical when $m>2$
(section \ref{sec:bif_m>2}). When $m=2$, the pitchfork bifurcation is
infinitely degenerate, we have $\tilde K - 2 = o(x_1^n)$ for all
$n\geq 2$ in a neighborhood of the bifurcation, 
where $x_1$ is the order parameter that characterizes equilibria in bifurcation diagrams,
$x_1=0$ corresponding to the uncoherent equilibrium and 
$x_1 >0$ to coherent equilibria. This bifurcation
scenario and the transition from supercritical to subcritical
pitchfork at $m=2$ are shown in section \ref{sec:bif_m=2}. The
bifurcation scenario in the limit $m \rightarrow + \infty$ is
discussed in section \ref{sec:bif_m_grand}. Several remarks and
discussions are in section \ref{sec:discussion}. Details on numerical
methods used to make figures of this text can be found in the
appendix.

\section{Positivity and mass conservation}
\label{sec:pos_mass_csv}

\par The problem (\ref{eq:main}) is reflexion invariant: if $u(t,\theta)$ is a solution of (\ref{eq:main}) then $v(t,\theta) = u(t, - \theta)$ is also a solution of 
(\ref{eq:main}) with $v_0(\theta) = u_0(-\theta)$. In particular, if $u_0$ is an even function then any solution $u(t, \cdot)$ of (\ref{eq:main}) is even. 
Here we focus only on even initial data and we suppose $u_0(\theta) = u_0(- \theta)$ in the following.
If $u_0 \in L^1(\mathbb S)$ with $u_0 \geq 0$, and $u(t,\cdot)$ is an associated solution of (\ref{eq:main}) on $[0,T]$, then we have 
$u(t, \cdot) \geq 0$ and $\int_\mathbb S u(t, \theta) d \theta = \int_\mathbb S u_0(\theta) d \theta $ for all $t \in [0,T]$ (see lemma \ref{lm:pos_mass_csv}).
In the following we assume $u_0 \geq 0$ and $ \int_{\mathbb S^1} u_0(\theta) d \theta = 1$.
We use the notation 
$$ u(t,\theta) = \frac{1}{2\pi} + \frac 1\pi \sum_{n \geq 1} x_n(t) \cos(n\theta),  $$
and in particular we denote $x_1 = x_1(u) = \int_{0}^{2\pi} u(\theta) \cos(\theta) d \theta $.

\begin{lm}\label{lm:pos_mass_csv}
\par Assume that $u_0 \in L^1$ and $u$ is a weak solution of (\ref{eq:main}) in $C^{1}([0,T],L^1([0,2\pi]))$. Then we have 
$\int_{-\pi}^{\pi} u(t, \theta) d \theta = \int_{-\pi}^{\pi} u_0( \theta) d \theta $ for all $t \in [0,T]$.
Suppose that $u \in C^{1}([0,T],L^1([0,2\pi]))$ is a weak solution of (\ref{eq:main}) with $u(0, \cdot) =u_0 \geq 0$. 
Then we have $u(t,\cdot) \geq 0$ for all times $t \in [0,T]$.
\end{lm}
\begin{proof}
\par For any time $t \in [0,T]$ and $\phi \in C^\infty([0,T]\times[0,2\pi])$, we have 
$$ \int_0^{2\pi}  u(t,\theta) \phi(\theta) d \theta = \int_0^{2\pi}  u_0(\theta) \phi(\theta) d \theta 
+ \int_0^T  \int_0^{2\pi} \left ( u(s,\theta)^m \partial_\theta^2 \phi(s,\theta) 
- u(s,\theta) J\ast u(s,\theta) \partial_\theta \phi(s,\theta \right )   d\theta     ds,   $$
and for $\phi(s,\theta)=1$ we find $\int_{-\pi}^{\pi} u(t, \theta) d \theta = \int_{-\pi}^{\pi} u_0( \theta) d \theta $ for all $t \leq T$.
\par Consider $u_-(t,\theta) = \max(- u(t,\theta), 0) \geq 0$. We have $u_- \in C^{1}([0,T],L^1([0,2\pi])) $ and $u_-$ is a weak solution of 
(\ref{eq:main}). Then the first part of the lemma gives $\int_0^{2\pi} u_-(t,\theta) d\theta = \int_0^{2\pi} u_-(0,\theta)d \theta =0$ for all $t \in [0,T]$. 
With $u_- \geq 0$ this implies $u(t, \cdot) = 0$ and then $u(t,\cdot) \geq 0$ for all $t \in [0,T]$.
\end{proof}

\section{Gradient Flows}
\label{sec:grad_flow}

\subsection{Energy decay and convergence to steady states}
\par In all this section, we assume $m > 1$. Equation (\ref{eq:main}) can be written as a gradient flow for the Wasserstein metric: taking 
\begin{equation}\label{eq:F} 
\mathcal F(u)= \frac{1}{m-1} \int_\mathcal S u^{m}(\theta) d \theta  
 - \frac K2 \iint_{\mathcal S \times \mathcal S} u(\theta) u(\varphi) \cos(\theta - \varphi) d\theta d \varphi , 
\end{equation}
we have formally
\begin{equation} 
  \frac{\delta \mathcal F}{ \delta u} = \frac{m}{m-1} u^{m-1} - K \int_\mathcal S \cos(\theta - \varphi) u (\varphi) d\varphi,
\end{equation} 
 and 
\begin{equation}
\partial_\theta \left( u \partial_\theta \frac{\delta \mathcal F}{\delta u(\theta)} \right) 
= \partial_\theta  \left ( m u u^{m-2} \partial_\theta u + K u \int_\mathcal S \sin(\cdot - \varphi) u (\varphi) d\varphi \right ),
\end{equation}
so that (\ref{eq:main}) is equivalent to 
\begin{equation}\label{eq:grad_flow}
\partial_t u = \partial_\theta \left( u \partial_\theta \frac{\delta \mathcal F}{\delta u(\theta)} \right). 
\end{equation}

\par The functional $ \mathcal F$ is a strict Lyapunov functional for (\ref{eq:main}). 
More precisely, if $u(t, \cdot)$ is a smooth non-constant solution of (\ref{eq:main}), we have 
\begin{equation} 
\begin{split}
\partial_t \mathcal F(u(t)) &=  D \mathcal F (u) . \partial_t u = \int_\mathcal S \frac{\delta \mathcal F}{\delta u(\theta)} . \partial_t u \,d \theta
= - \int_\mathcal S u(t,\theta) \left [ \partial_\theta \frac{\delta \mathcal F}{\delta u(\theta)} (t,\theta)\right ]^2 d \theta  \\
& = -  \int_\mathcal S u  \left ( m u^{m-2} \partial_\theta u + K x_1 \sin(\cdot) \right )^2 d \theta <0
\end{split}
\end{equation}
The solutions $\hat u \in L^\infty$ with $\partial_\theta (\hat u^{m-1}) \in L^\infty$ of 
\begin{equation}\label{eq:eq_equilibria}
\hat u(\theta) \left(  \frac{m}{m-1} \partial_\theta ( \hat u^{m-1})(\theta) + K \hat x_1 \sin(\theta) \right) = 0, 
\end{equation}
where $ \hat x_1 = \int_\mathcal S \hat u(\theta) \cos(\theta) d \theta$,
are the equilibria of equation (\ref{eq:main}) (see section \ref{sec_equilibria} for a description of the set of solutions of (\ref{eq:eq_equilibria})). 

\begin{prop}\label{prop:reg_F}
The functional $\mathcal F : L^m(\mathbb S) \rightarrow \mathbb R$ defined in (\ref{eq:F}) is $C^1$.
If $u \in L^m(\mathbb S)$ with $u \geq 0$ and $\int_\mathbb S u \, d\theta = 1$, then we have 
\begin{equation}
- \frac K2 \leq  \frac{m}{m-1} \Vert u \Vert_{L^m}^m  - \frac K2 \leq  \mathcal F(u) \leq \frac{m}{m-1} \Vert u \Vert_{L^m}^m + \frac K2
\end{equation}
\end{prop}
\begin{proof}
The regularity of $\mathcal F$ is a direct consequence of its definition (\ref{eq:F}). 
The hypotheses $u \geq 0$ and $\int u d\theta = 1$ imply in particular $0 \leq x_1(u) \leq 1$, 
and the inequalities on $\mathcal F(u)$ follow. 
\end{proof}

\begin{thm}\label{thm:cv}
Assume that $u$ is a solution of (\ref{eq:main}) with $u(t, \cdot) \geq 0$ and $\int_\mathcal S u(t,\theta) d\theta = 1$ for all $t \geq 0$, 
and $u \in C^{1}([0, + \infty[, L^m(\mathbb S))$. 
Let $C$ be the set of solutions of (\ref{eq:eq_equilibria}). 
\par Then we have 
\begin{equation}\label{eq:cv_equil}
\text{dist}_{L^\infty}(u(t), C) = \underset{\hat u \in C}{\inf}  \Vert u(t) - \hat u \Vert_{L^\infty}  \underset{t \rightarrow + \infty}{\longrightarrow} 0.
\end{equation}
\end{thm}

\par { \em Remark. } Lemma \ref{lm:pos_mass_csv} and theorem \ref{thm:cv} directly extend to initial data with
$u_0 \geq 0$ and $\int_{-\pi}^{\pi} u_0(\theta) d \theta < + \infty$ only. In particular, solutions of (\ref{eq:main}) converge to 
its equilibria set even for non-even initial conditions. 
This also holds when $m=1$, see \cite{MR2594897}.

\begin{proof}
\par For all $t \geq 0$ we have 
$ \Vert u(t) \Vert_{L^m}^m \leq \mathcal F(u(t)) + \frac K2 \leq \mathcal F(u(0)) + \frac K2 $, and we have $u \in L^\infty([0, + \infty[, L^m) $ 
or equivalently $u^m \in L^\infty([0, + \infty[, L^1)$.
By hypothesis and proposition \ref{prop:reg_F}, the function $t \mapsto \mathcal F \circ u (t)$ is  $C^1([0, + \infty[, \mathbb R)$, 
not increasing and bounded below. 
Hence we have $\mathcal F(u(t)) \rightarrow l \geq - \frac K2$ and $\partial_t \mathcal F (u(t)) \rightarrow 0$ when $t \rightarrow + \infty$. 
\par Let $t_n \geq 0$ be a diverging sequence and $u_n = u(t_n, \cdot)$.
The sequence $(u_n)$ is bounded in $L^m$ and (up to taking a subsequence), we can assume that $u_n \stackrel{w L^m}{\longrightarrow} u_\infty$
(where $\stackrel{w L^p}{\longrightarrow}$ denotes weak convergence in $L^p$ space).
We also have 
$$\int_\mathcal S u_n  \left ( m u_n^{m-2} \partial_\theta u + K x_1(u_n) \sin(\cdot) \right )^2 d \theta \longrightarrow 0,$$
and $ m u_n^{m- \frac 32} \partial_\theta u + K x_1(u_n)  \sqrt{u_n} \sin(\cdot) $ is convergent in $L^2(\mathbb S)$.
We have $u_n \stackrel{w L^1}{\longrightarrow} u_\infty$, so that $x_1(u_n) \rightarrow x_1(u_\infty)$, and 
$$ K x_1(u_n) \sqrt{u_n} \sin(\cdot) \stackrel{w L^{2m}}{\longrightarrow} K x_1(u_\infty) \sqrt{u_\infty} \sin(\cdot). $$ 
In particular $m u_n^{m - 3/2} \partial_\theta u_n = \frac{m}{m- \frac 12} \partial_\theta \left ( u_n^{m - 1/2} \right )$ 
is bounded in $L^{2m} \subset L^2$. Then $u_n^{m - 1/2}$ is bounded in $H^1$ and compact in $L^\infty$. 
Since $u_n^{m- 1/2} {\longrightarrow} u_\infty^{m-1/2}$ weakly in $L^{\frac{m}{m - 1/2}}$, 
we have $u_n^{m- 1/2} {\longrightarrow} u_\infty^{m-1/2}$ in $L^\infty$, and we deduce $u_n {\longrightarrow} u_\infty$ in $L^\infty$.
The limit $u_\infty \in L^\infty$ satisfies 
$\partial_\theta \left [ u_\infty \left( \frac{m}{m-1} \partial_\theta u_\infty^{m-1} +K x_1 \sin(\theta)  \right) \right ] = 0 $ 
at least in the sense of distributions. Since $\partial_\theta u_\infty^{m-1} \in L^\infty$, we have $u_\infty \in C$. 
\par We have shown that the limit $u_\infty$ of any converging subsequence $u(t_n, \cdot)$ is in the equilibria set $C$, 
and (\ref{eq:cv_equil}) follows.
\end{proof}

%

\subsection{Linearization at coherent equilibria}

\par Recall that we consider equation (\ref{eq:main}) in a functions space with $u \geq 0$, 
$\int_\mathbb S u \,d\theta =1$ and $u$ is even: $u(\theta) = u(-\theta)$.
In this section, $u$ denotes an equilibrium of equation (\ref{eq:main}).

\medskip
\par Consider $v$, $w$ such that $\int_{\mathbb S} v d\theta = \int_{\mathbb S} w d\theta  =0$. 
Suppose that they are $V$ and $W$ such that 
\begin{equation}
 v= \partial_\theta ( u \partial_\theta V ) \;\;\text{ and } \;\; w= \partial_\theta ( u \partial_\theta W ). 
\end{equation}
Then we have (see \cite{MR2192294,MR1842429} for example)
$$ \int_\mathbb S v W d\theta = \int_\mathbb S V w d\theta = - \int_\mathbb S u \partial_\theta V \partial_\theta W d \theta. $$
For such $v$ and $w$, we define the bilinear form $(v,w) = - \int_\mathbb S v W d\theta = - \int_\mathbb S V w d\theta$ 
(note that $(\,,)$ depends on $u$). 
For such a $v \neq 0$ we have $(v,v) > 0$, and then $v,w \mapsto (v,w)$ defines a scalar product. 

\begin{prop}\label{prop:Lu_symmetric}
\par Let $L_u$ be the operator associated to the linearization of (\ref{eq:main}) at the equilibrium $u$, 
$$ L_u v = \partial_\theta \left( m u^{m-1}\partial_\theta v +  v J\ast u + u J\ast v  \right ).$$
Let $v$, $w$ be even and chosen as above. Then we have 
\begin{equation}\label{eq:op_sym}
(L_u v, w) = -  m  \int_\mathbb S u^{m-2}(\theta) v(\theta) w(\theta) d\theta  + \int_\mathbb S \tilde J\ast v(\theta) w(\theta) d\theta = (v, L_u w), 
\end{equation}
where $\tilde J (\theta)= K \cos(\theta)$.
In particular, the operator $L_u$ is symmetric for the scalar product $(\,,)$.
\end{prop}
\begin{proof}
We have 
$$ (L_u v, w) =  \int_\mathbb S  \left( m u^{m-1}\partial_\theta v +  v J\ast u + u J\ast v  \right ) \partial_\theta W d \theta.$$ 
With $ \partial_\theta (u^m) + u J\ast u=0 $, this gives 
\begin{equation}
\begin{split}
 (L_u v, w) &=  \int_\mathbb S  \left( m u^{m-1}\partial_\theta \left( \frac 1u v \right) +  J\ast v  \right ) u \partial_\theta W d \theta \\
 &=  -  m  \int_\mathbb S u^{m-2} v w d\theta  - m \int_\mathbb S \partial_\theta(u^{m-1}) v \partial_\theta W d\theta
 - \int_\mathbb S \tilde J\ast v  w d \theta.
\end{split}
\end{equation}
Since $u^{m-1}$, $v$, and $\partial_\theta W$ are even, we have $\int_{\mathbb S} \partial_\theta(u^{m-1}) v \partial_\theta W d\theta = 0$
and the first equality of (\ref{eq:op_sym}) follows.
We have 
$$\int \tilde J\ast v(\theta) w(\theta) d\theta = \iint K \cos(\theta - \varphi) v(\varphi) w(\theta) d \varphi d \theta,$$
hence the second term of (\ref{eq:op_sym}) is symmetric in $v$ and $w$, and we have $(L_u v, w) = (v, L_u w)$.
\end{proof}

\par The importance of the above result lies in that self-adjoint operators with compact resolvent have real pure point countable spectra, and 
they are diagonalizable in a Hilbert eigenbasis. 
The existence of spectral gap is then sufficent to show hyperbolicity and linear stability of equilibria, 
and classical stable, central and unstable manifolds existence theorem are available. 
This allows for the definition of an approriate projections and
 the study of small perturbations effects on the dynamics in neighborhoods of the equilibria set 
(see \cite{MR2594897,GPP,GPPP} in case $m=1$).

\section{Structure of the set of equilibria}
\label{sec_equilibria}

To serve as a basis for comparison, we recall briefly the bifurcation scenario in case $m=1$ 
(see \cite{LuoZZ_2004,MR2109485,silver:468,MR2594897,GPP}). 
When the coupling strength is below a critical value $K_c$, equation (\ref{eq:main}) has a unique (stable) equilibrium $\frac{1}{2\pi}$
that is called uncoherent. 
When $K=K_c$ a supercritical pitchfork bifurcation occurs, and for all $K > K_c$ there are exactly three equilibria: $\frac1{2\pi}$
and $u(\theta) = \frac1c e^{\alpha \cos(\theta)}$, where $c = \int_{-\pi}^\pi e^{\alpha \cos(\theta)} d\theta$ and 
$\alpha = \int_{-\pi}^{\pi} u(\theta) \cos(\theta) d\theta$, and $u(\cdot + \pi)$. 
When $K > K_c$, this equilibrium $u$ is locally stable and attracts all solutions of (\ref{eq:main})
 such that $\int_{-\pi}^\pi u_0 (\theta) \cos(\theta) d \theta > 0$, whereas all solutions with 
 $\int_{-\pi}^\pi u_0 (\theta) \cos(\theta) d \theta < 0$ converge to $u(\cdot + \pi)$. 
These two equilibria are non constant, they have a unique maximum on $[-\pi, \pi]$ and they are called coherent in that sense. 
In particular, coherent equilibria have been shown to diverge from $\frac{1}{2\pi}$ 
as $(K - K_c)^\frac 12$ (\cite{RevModPhys.77.137,MR1115806,Crawford19991} and references therein). 

In this section, we establish the main elements of proof of theorem \ref{thm:main}. 
In the next section \ref{sec:global_desc} we detail further bifurcations scenario (depending on $m$), 
which lead to the phase transition diagram of figure \ref{fig:bif_UBC}.

\subsection{The uncoherent equilibrium $\frac{1}{2\pi}$}
\label{sec_equilibria_uncoherent}

A pitchfork bifurcation occurs at the uncoherent equilibrium when $K = K_b =2m\left( \frac{1}{2\pi} \right)^{m-1} $. 
The uncoherent equilibrium is linearly stable when $K < K_b$ and unstable when $K > K_b$. 
In the following we rescale the parameter
$ \tilde K = K \frac{1}{m} \left( \frac{1}{2\pi} \right)^{-(m-1)}, $ so that the bifurcation occurs at 
$\tilde K = 2$ independently of $m$. 

\begin{prop}
\par For all $m >0$ and $K \geq 0$, $u (\theta) = \frac{1}{2\pi}$ is an equilibrium of equation (\ref{eq:main}). 
The linearization of equation (\ref{eq:main}) at $\frac{1}{2\pi}$ is: 
\begin{equation}
\partial_t v =L_{\frac{1}{2\pi}} v = m \left ( \frac{1}{2\pi} \right )^{m-1} \partial_\theta^2 v 
+ \frac{K}{2\pi} x_1 \cos(\theta), 
\end{equation}
the operator $L_{\frac{1}{2\pi}}$ is orthogonal is the classical Fourier basis of $L^2(\mathbb S)$ 
and its spectrum is given by $\lambda_1 = \frac{K}{2} - 2m\left( \frac{1}{2\pi} \right)^{m-1} $ and 
$\lambda_n = -n^2 2m\left( \frac{1}{2\pi} \right)^{m-1}  $ for all $n \geq 2$.
\par In particular, $\frac{1}{2\pi}$ is linearly stable if and only if 
\begin{equation}
K <  K_b= 2m\left( \frac{1}{2\pi} \right)^{m-1}.
\end{equation}
For $K = K_b$ the operator $L_{\frac{1}{2\pi}}$ has a null eigenvalue, 
and for $K > K_b$ the equilibrium $\frac{1}{2\pi}$ is linearly unstable. 
\end{prop}

\subsection{Coherent equilibria of equation (\ref{eq:main})}
\label{sec_equilibria_coherent}

\begin{prop}\label{prop:eq_coherent_eq} 
\par Suppose $m > 1$. Equilibria $\hat u$ of equation (\ref{eq:main}) satisfy either $\hat u (\theta) = \frac{1}{2 \pi}$ or
\begin{equation}\label{eq:equilibria}
 \hat u^{m-1}= \frac{m-1}{m} K x_1 \left [ \cos (\theta) + c \right ]_+,  
\end{equation}
with $\int_{-\pi}^\pi \hat u(\theta) d \theta = 1$ and $\int_{- \pi}^\pi \hat u(\theta) \cos (\theta) d \theta = x_1$ and $c > -1$,
where $x_+$ denotes $\max(x,0)$. 
\par For every $c > -1$, there is a unique $K > 0$ and $x_1 \in ]0, 1]$ such that (\ref{eq:equilibria}) holds, they are given by 
\begin{equation}\label{eq:K_x_1}
K = \frac{m}{m-1} \frac{1}{ J_m(c)} \frac{1}{I_m(c)^{m-2}} \;\text{ and } \; x_1 = \frac{J_m(c)}{I_m(c)}, 
\end{equation}
where $I_m(c) = \int_{-\pi}^\pi \left [ \cos(\theta) + c \right ]_+^{\frac{1}{m-1}} d \theta $ 
and $J_m(c) = \int_{-\pi}^\pi \left [ \cos(\theta) + c \right ]_+^{\frac{1}{m-1}} \cos(\theta) d \theta $. 
\par In case $0<m<1$, the non-constant equilibria of equation (\ref{eq:main}) are 
\begin{equation}\label{eq:equilibria_m<1}
 \hat u^{m-1}= \frac{1-m}{m} K x_1 \left [ c- \cos (\theta) \right ]_+,  
\end{equation}
with $c >1$, where $x_1$ and $K$ are given by 
$$x_1 = \frac{\mathcal J_m(c)}{ \mathcal I_m(c)} \; \text{ and } \; K = \frac{m}{1-m} \frac{1}{ \mathcal  J_m(c)} \mathcal  I_m(c)^{2-m} $$
with $ \mathcal I_m(c) = \int_{-\pi}^\pi \left [ c- \cos(\theta)  \right ]_+^{\frac{-1}{1-m}} d \theta  $
and $\mathcal J_m(c) = \int_{-\pi}^\pi \left [ c- \cos(\theta) \right ]_+^{\frac{-1}{1-m}} \cos(\theta) d \theta $. 
\end{prop}
\begin{proof}
Assume that $m>1$ and $u$ is a non-constant equilibrium of (\ref{eq:main}).
Following equation (\ref{eq:eq_equilibria}), for all $\alpha \in [0, 2\pi]$ we have either $u(\alpha)=0$ or 
$u^{m-1}(\theta)= \frac{K(m-1)}{m} x_1  \left( \cos (\theta) + c \right) $ on an open interval containing $\alpha$. 
With $u^{m-1} \in W^{1, \infty}(\mathbb S)$ and $u \geq 0$, this implies 
$u^{m-1}(\theta)= \frac{K(m-1)}{m} x_1  \left( \cos (\theta) + c \right)_+ $ on $[- \arccos(-c), \arccos(-c)]$, 
that is an increasing function of $\theta$.
Thus there is a unique $c$ such that (\ref{eq:equilibria}) holds on $\text{Supp}(u)=[- \arccos(-c), \arccos(-c)]$ and 
$u^{m-1}(\theta)=0$ on $[0, 2 \pi] \backslash \text{Supp}(u) $.
\par We have 
\begin{equation}\label{eq:cond_u_1}
\int_{-\pi}^{\pi} u(\theta) d\theta = \left( \frac{K(m-1)}{m}x_1 \right)^{\frac{1}{m-1}} I_m(c) = 1,
\end{equation}
\begin{equation}\label{eq:cond_u_x1}
\int_{-\pi}^{\pi} u(\theta) \cos(\theta) d\theta = \left( \frac{K(m-1)}{m}x_1 \right)^{\frac{1}{m-1}} J_m(c) = x_1,
\end{equation}
which give directly $x_1 = \frac{J_m}{I_m}$. 
Then we have $ K\frac{m-1}{m} J_m(c)^{m-1} = x_1^{m-2} $ and (\ref{eq:K_x_1}) follows.
\par In case $0<m<1$, equation (\ref{eq:equilibria_m<1}) is proved similarly. 
\end{proof}

\par Proposition \ref{prop:pitchfork_bif} shows in particular  
that the branch of equilibria of proposition \ref{prop:eq_coherent_eq} appears by a pitchfork bifurcation at $\frac1{2\pi}$ when $K = K_b$,
 or equivalently $\tilde K = 2$, and that it persists until $K \rightarrow + \infty$.
Equation (\ref{eq:amplit_bif_fourche}) shows that when $1<m<2$ the pitchfork bifurcation is supercritical, and when $m>2$ the pitchfork bifurcation 
is subcritical. 
In both cases $1<m<2$ and $m>2$, the distance between coherent equilibria and the uncoherent equilibrium scales as $\left ( \tilde K - 2 \right )^{1/2}$ in 
a neighborhood of the bifurcation. 
When $m=2$ the pitchfork bifurcation degenerates, the distance between coherent equilibria and the uncoherent equilibrium is larger 
than some positive constant for any $\tilde K >2$  (see also proposition \ref{prop:eq_coherent_eq_m2} for the case $m=2$ ).

\begin{prop}\label{prop:eq_coherent_eq_m2}
\par In the case $m=2$, (equations (\ref{eq:K_x_1}) still hold) the equilibria are given by  
\begin{itemize}
\item for $c \geq 1$, $$K = \frac{2}{\pi} \;\text{ and } \; x_1 = \frac{1}{2c},$$ 
\item for $-1< c < 1$, $$K = \frac{2}{\arccos(-c)} \;\text{ and } \; x_1 = \frac{\arccos(-c)}{2c \arccos(-c) + 2 \sqrt{1 - c^2}}.$$ 
\end{itemize}
\end{prop}
\begin{proof}
In case $m=2$, their respective definitions give $I_2(c)= 2c \arccos(-c) + 2 \sqrt{1 - c^2} $  for $\vert c \vert < 1$ and 
$I_2(c) = \frac1{2c}$ for $c \geq 1$, and $J_2(c) = \arccos(-c)$ for $\vert c \vert <1$ and $J_2(c) = \pi$ for all $c \geq 1$. 
When $m=2$ equation (\ref{eq:cond_u_x1}) reads $ 2 K = J_m(c) $, and combined with $x_1 = \frac{J_m}{I_m}$,
 this implies property \ref{prop:eq_coherent_eq_m2}.
\end{proof}

\begin{figure}[h!tp]
\psfrag{Kt}[B][B][1][0]{  {$\tilde K$}}
\psfrag{c}[B][B][1][0]{  {$c$}}
\begin{center}
\includegraphics[scale=1.0]{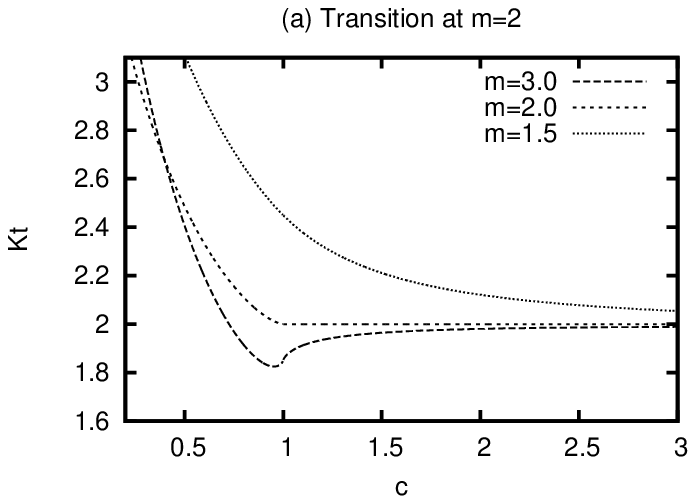}
\includegraphics[scale=1.0]{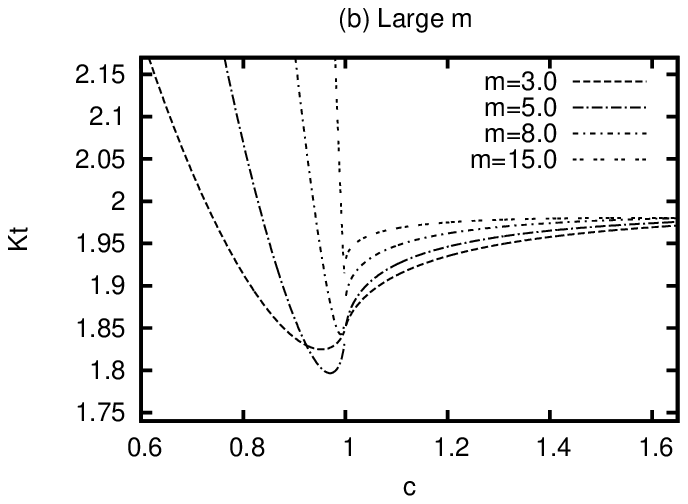}
\end{center}
\begin{center}
\vspace{-2mm}
\caption{$\tilde K$ as a function of $c$ for $m>1$. 
\small{ Coherent equilibria of (\ref{eq:main}) are uniquely determined by $\tilde K$ and $x_1$, which depend on the parameter $c \in ]-1; + \infty[$. 
When $1<m<2$, $ \tilde K_m(c)$ is decreasing on $]-1; +\infty[$,  
and for every $\tilde K \in [2; + \infty[$, there is exactly one coherent equilibrium of (\ref{eq:main}).
This is illustrated by the top most curve on the panel (a) for $m=1.5$.
For $m>2$, $\tilde K_m(c)$ has a minimum at some $c^* < 1$, it is decreasing on $]-1,c^*]$ and increasing on $[c^* ; + \infty[$.
The below most curve on panel (a) and the the four curves on panel (b) show that this holds for various values of $K=3.0$, 
$K=5.0$, $K=8.0$, and $K=15.0$.
For all $m$, we have $\tilde K \rightarrow 2$ when $c \rightarrow + \infty$.
When $m>2$, for all $\tilde K \in [\tilde K_m(c^*); 2]$ equation (\ref{eq:main}) has two coherent equilibria. 
They appear by a fold bifurcation when $\tilde K = \tilde K_m(c^*)$, and the unstable one disappears by subcritical
 pitchfork bifurcation at $\frac{1}{2\pi}$ when $\tilde K =2$ 
(see proposition \ref{prop:eq_coherent_eq} and illustrations in figure \ref{fig:K_c}) 
 }
}\label{fig:K_c}
\end{center}
\end{figure}

\begin{prop}\label{prop:pitchfork_bif}
Assume that $m>1$.
\begin{itemize}
\item We have $ \vert x_1 \vert \leq 1 $ and $K > 0$ for all $c \in ]-1; + \infty[$.
\item We have 
$x_1 = 1 - \beta_m (1+c) +o(1+c) $ and 
\begin{equation}\label{eq:asym_c-1}
 K = \frac{m}{m-1} \gamma_m^{1-m} (1+c)^{- \frac{m+1}{2}}  + o(1+c)^{- \frac{m+1}{2}},
\end{equation} when $c \rightarrow -1$, where the constants $\beta_m >0$ and $\gamma_m >0$ are given in (\ref{eq:beta_gamma}). 
In particular we have 
\begin{equation}\label{eq:x1_Kgrand}
x_1 = 1 - \mu K^{- \frac2{m+1}} + o\left ( K^{- \frac2{m+1}} \right) \; \text{ when } \; K\rightarrow + \infty, 
\end{equation} with $\mu = \beta_m \left ( \frac{m-1}{m}\right )^{-\frac{2}{m+1}} \gamma_m^{-2 \frac{m-1}{m+1}}  >0$.
\item When $c \rightarrow + \infty$, we have
\begin{equation}\label{eq:asym_cgrand_x1}
x_1 = \frac1{2(m-1)} \frac 1c \left [ 1 + \frac{(2-m)(1-2m)}{8(m-1)^2} \frac1{c^2} + o \left( \frac1{c^2} \right)   \right ],  
\end{equation}
\begin{equation}\label{eq:asym_cgrand_K}
K = 2 m \left( 2 \pi \right)^{1-m} \left[ 1 - \frac{2-m}{8(m-1)^2} \frac1{c^2} + o \left ( \frac1{c^2} \right)  \right],
\end{equation} 
and in particular 
\begin{equation}\label{eq:amplit_bif_fourche}
 \tilde K - 2 = - (m-2) x_1^2 + o \left( x_1^2 \right).
\end{equation}

\end{itemize}
\end{prop}
\begin{proof}
\par We have $I_m(c) > 0$ for all $c > -1$, and since $\theta \rightarrow \left(c + \cos(\theta) \right)_+ $ 
is not increasing and non negative on $[0, \pi]$, we have $J_m(c) \geq 0$. 
Furthermore we have $ \vert J_m(c) \vert \leq I_m(c) $ for all $c>-1$, and the inequalities $\vert x_1 \vert \leq 1$ and $K > 0$ follow. 
\par Denoting $\epsilon = 1+c$, we have
\begin{equation}
\begin{split}
I_m(c) &= 2 \int_{0}^{\arccos(-c)} \left( c + \cos(\theta) \right )^{\frac1{m-1}} d\theta
       = 2 \int_{-c}^1 \left(  c + y  \right )^{\frac1{m-1}} \frac{1}{\sqrt{1-y^2}} dy \\
       &= 2 \sqrt{\epsilon}\, \epsilon^{\frac1{m-1}} \int_0^1 \left(  1 - \lambda  \right )^{\frac1{m-1}} \frac{1}{\sqrt{\lambda(2- \epsilon \lambda)}} d \lambda,
\end{split}
\end{equation}  
and similarly 
\begin{equation}
J_m(c) = 2 \sqrt{\epsilon}\, \epsilon^{\frac1{m-1}} \int_0^1 \left(  1 - \lambda  \right )^{\frac1{m-1}}
 \frac{1 - \epsilon \lambda}{\sqrt{\lambda(2- \epsilon \lambda)}} d \lambda.
\end{equation}
This implies $x_1 \rightarrow 1$ in (\ref{eq:asym_c-1}), and with $K = \frac{m}{m-1} \frac{1}{x_1} \frac{1}{(I_m)^{m-1}} $ 
we also have $K \rightarrow + \infty$ when $c \rightarrow -1$.
More precisely, we have 
\begin{equation}
x_1 = 1 - \beta_m \epsilon +o(\epsilon) \;\text{ and } \;
 K = \frac{m}{m-1} \gamma_m^{1-m} \epsilon^{- \frac{m+1}{2}} \left( 1 - \beta_m \epsilon + o(\epsilon) \right),
\end{equation}
where 
\begin{equation}\label{eq:beta_gamma}
 \gamma_m = 2 \int_0^1 (1 - \lambda)^{\frac1{m-1}} \frac1{\sqrt{\lambda(2-\lambda)}} d \lambda 
\; \text{ and }\; 
\beta = \frac{2}{\gamma} \int_0^1 (1 - \lambda)^{\frac1{m-1}} \frac \lambda{\sqrt{\lambda(2-\lambda)}} d \lambda,  
\end{equation}
and this implies (\ref{eq:asym_c-1}) and (\ref{eq:x1_Kgrand}).
\par Using the Taylor expansion $(1+x)^\alpha = 1 + \alpha x + \frac{\alpha(\alpha-1)}{2}x^2 + \frac{\alpha(\alpha-1)(\alpha-2)}{6}x^3 +  O(x^4)$, we find 
\begin{equation}
 I_m(c) = c^{\frac{1}{m-1}} \int_{-\pi}^{\pi} \left ( 1 + \frac 1c \cos(\theta) \right )^{\frac1{m-1}} d \theta 
=2\pi c^{\frac1{m-1}} \left [ 1 +  \frac{2-m}{4(m-1)^2} \frac1{c^2} + O\left ( \frac1{c^4} \right ) \right ]. 
\end{equation} when $ c \rightarrow + \infty$.
Similarly, we have
\begin{equation}
J_m(c) = \pi c^{\frac{1}{m-1}} \frac{1}{m-1} \frac1c \left ( 1 + \frac{(3-2m)(2-m)}{8(m-1)^2} \frac1{c^2} + O\left( \frac1{c^4} \right)  \right ).
\end{equation}
The estimate (\ref{eq:asym_cgrand_x1}) on $x_1$ is a consequence of (\ref{eq:K_x_1}) and the two estimates above.
The asymptotic of $K$ is a consequence of this estimate on $x_1$ and 
$ K = \frac{m}{m-1} \frac 1{x_1} \frac 1 {(I_m(c))^{m-1}}$.  Equation (\ref{eq:amplit_bif_fourche})
 is a consequence of (\ref{eq:asym_cgrand_x1}) and (\ref{eq:asym_cgrand_K})
\end{proof}

\par The bifurcations in cases $1<m<2$, $m=2$, $m>2$ and the transition at $m=2$ are discussed in the following sections.
The case $m \rightarrow + \infty$ is also discussed in section \ref{sec:bif_m_grand}.

\section{Global description of results}
\label{sec:global_desc}

\par After this analysis of the equilibria set of equation (\ref{eq:main}),
we synthesize our findings by presenting, for various ranges of the
nonlinear diffusion coefficient $m$, the corresponding bifurcation
diagrams when the parameter $\tilde K$ is varied. One of the main
ingredients in our analysis of equilibria and their dependence on
parameters has been the function $c\rightarrow \tilde K$. It is the
shape of this function that determines the number of non trivial
equilibria of the system. The left panel in figure 2 illustrates how
this function changes for various values of the parameter $m$. For
$1<m<2$, it is strictly decreasing while for $m>2$ it exhibits a
unique minimum. In between, at $m=2$, the function is strictly
decreasing with a zero slope at its right hand limit. The right panel
of the same figure shows how the general shape and the minimal value
of the function change when $m$ is further increased. In the following
paragraphs, we describe the dynamics of the system for these three cases of $m$
and illustrate through numerical computations of
representative examples how this change in the shape of $ c\rightarrow
\tilde K $ translates into the properties of equation (\ref{eq:main}).

\subsection{Coherence transition for $1<m<2$}
\label{sec:bif_m<2}

\par When $1<m<2$ the function $\tilde K = \tilde K(c)$ is decreasing on $]-1, + \infty[$ (see illustration figure \ref{fig:K_c} panel (a) for $K=1.5$), and
the function $x_1 = x_1(c)$ is decreasing on $]-1, + \infty[$. 
The bifurcation diagram $x_1 = x_1(\tilde K)$ is shown in figure \ref{fig:bif_m15} panel (a).

\par When $ \tilde K < \tilde K_c=2$ the unique equilibrium of (\ref{eq:main}) is the uncoherent steady state $\frac{1}{2\pi}$ 
(corresponding to $x_1=0$), which is locally and globally stable (label ($P_1$) figure \ref{fig:bif_m15} panel (a)). 
At $\tilde K= \tilde K_c = 2$ a supercritical pitchfork bifurcation occurs (see proposition \ref{prop:pitchfork_bif}). 
When $\tilde K >2 $, the uncoherent steady state $\frac{1}{2\pi}$ is locally unstable (label ($P_2$) figure \ref{fig:bif_m15} panel (a))
and equation (\ref{eq:main}) has a coherent equilibrium with $x_1 >0$. When $\tilde K < \tilde K(c=1) \approx 2.4$ this coherent equilibrium is positive 
on $[-\pi, \pi]$ (label ($S_p$) figure \ref{fig:bif_m15} panel (a)), whereas when $\tilde K > \tilde K(c=1)$ this coherent equilibrium is localized:
its support is strictly included in $]-\pi, \pi[$ (label ($S_l$) figure \ref{fig:bif_m15} panel (a)). 
\par When $1<m<2$, solutions $u(t,\theta)$ of (\ref{eq:main}) are smooth for all times $t \geq 0$.
When $\tilde K< \tilde K_c$, all solutions converge to the uncoherent equilibrium $\frac{1}{2\pi}$. 
Figure \ref{fig:cv_m15} panels (d) and (e) show one such converging solution for $\tilde K=1.5 $.

When $\tilde K> \tilde K(c=1)$ most solutions converge to the coherent equilibrium $(S_l)$ of (\ref{eq:main})
with localized support (see illustration in figure \ref{fig:cv_m15} panel (f) for $\tilde K = 6.0$). 

\begin{figure}[h!]
\begin{center}
\psfrag{K(tilde)}[B][B][1][0]{  {$\tilde K$}}
\psfrag{x1}[B][B][1][0]{  {$x_1$}}
\psfrag{theta}[B][B][1][0]{ \small {$\theta$}}
\psfrag{u(theta)}[B][B][1][0]{ \small {$u(\theta)$}}
\psfrag{1.0/6.28}[B][B][1][0]{ \tiny{$\frac{1}{2\pi}$}}
\psfrag{12pi}[B][B][1][0]{ \tiny{$\frac{1}{2\pi}$}}
\psfrag{Kt=2.4**}[B][B][1][0]{ \tiny{$\tilde K=2.4$}}
\psfrag{Kt=2.12**}[B][B][1][0]{ \tiny{$\tilde K=2.12$}}
\psfrag{Kt=2.01**}[B][B][1][0]{ \tiny{$\tilde K=2.01$}}
\psfrag{Kt=2.5**}[B][B][1][0]{ \tiny{$\tilde K=2.5$}}
\psfrag{Kt=13.5**}[B][B][1][0]{ \tiny{$\tilde K=13.5$}}
\psfrag{Kt=6.0**}[B][B][1][0]{ \tiny{$\tilde K=6.0$}}
\psfrag{u(t,theta)}[B][B][1][0]{ \small {$u(t,\theta)$}}
\psfrag{t=0.0**}[B][B][1][0]{ \tiny{$t=0.0$}}
\psfrag{t=1.0**}[B][B][1][0]{ \tiny{$t=1.0$}}
\psfrag{t=5.0**}[B][B][1][0]{ \tiny{$t=5.0$}}
\psfrag{t=8.0**}[B][B][1][0]{ \tiny{$t=8.0$}}
\psfrag{t=10.0**}[B][B][1][0]{ \tiny{$t=10.0$}}
\psfrag{t=15.0**}[B][B][1][0]{ \tiny{$t=15.0$}}
\psfrag{m=1.5}[B][B][1][0]{ \tiny {$m=1.5$}}
\psfrag{K=1.5}[B][B][1][0]{ \tiny {$\tilde K=1.5$}}
\psfrag{K=2.12}[B][B][1][0]{ \tiny {$\tilde K=2.12$}}
\psfrag{K=6.0}[B][B][1][0]{ \tiny {$\tilde K=6.0$}}
\begin{center}
\includegraphics[scale=0.7]{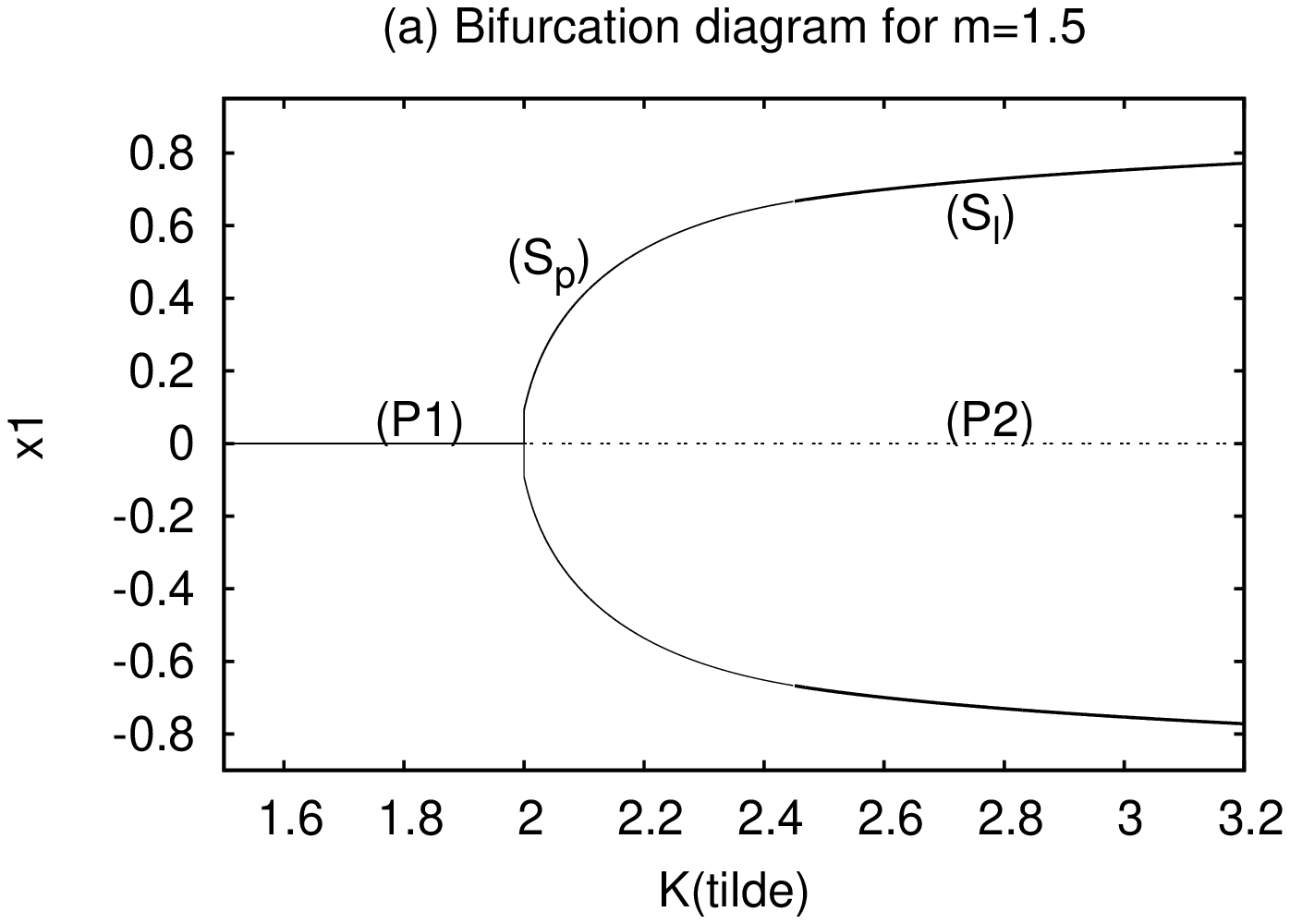}
\end{center}
\vspace{-4mm}
\begin{center}
\includegraphics[scale=0.9]{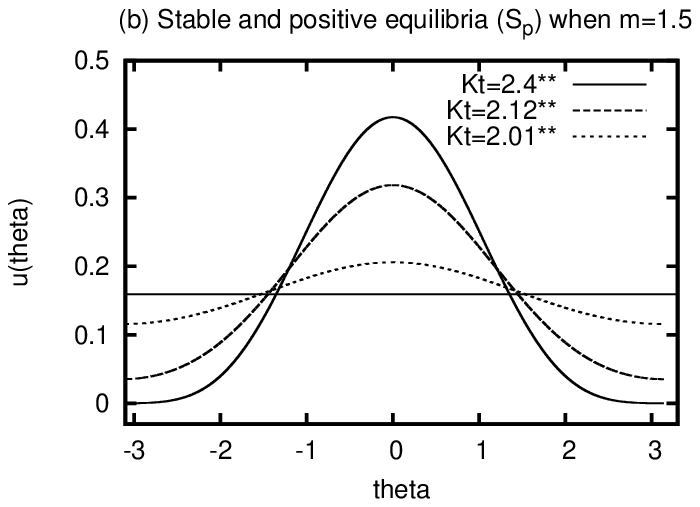}
\includegraphics[scale=0.9]{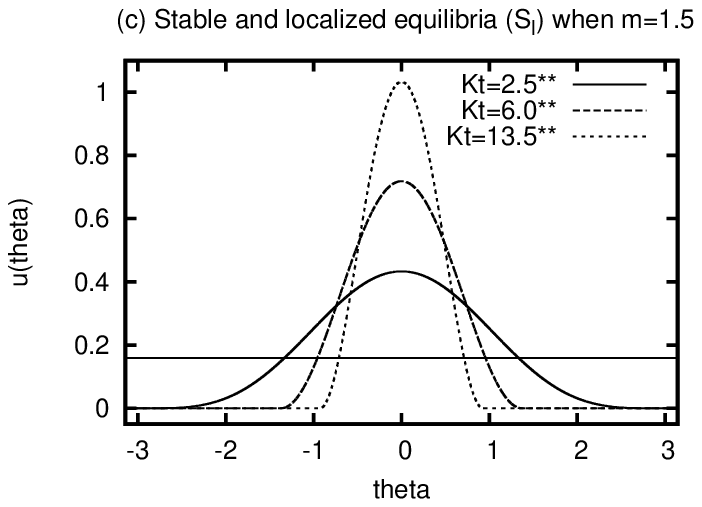}
\end{center}
\end{center}
\begin{center}
\vspace{-2mm}
\caption{Equilibria bifurcation diagram and equilibria of (\ref{eq:main}) in case $m=1.5$.
\small{
Panel (a) show the number of equilibria of (\ref{eq:main}) for $0 < \tilde  K < + \infty$.
We see in particular the supercritical pitchfork bifurcation occurring at $\frac1{2\pi}$ when $\tilde K = 2$. 
The branch of stable equilibria that appear satisfies 
$ \tilde K - \tilde K_c \sim -(m-2) x_1^2 > 0 $ in a neighborhood of the bifurcation (see proposition \ref{prop:pitchfork_bif}).
In panels (b) and (c) the horizontal line represents $\frac1{2\pi}$. 
Panel (b) shows the shape of coherent (stable) equilibria $(S_p)$, for various values of $\tilde K$ in  $]2, \tilde K(c=1) [$.
These equilibria emerge from $\frac1{2\pi}$ when $\tilde K = 2$ and they remains positive until $ \tilde K = \tilde K(c=1)$.
Panel (c) shows the shape of localized coherent equilibria $(S_l)$, for various values of $\tilde K > \tilde K(c=1)$. 
When $\tilde K = \tilde K(c=1)$ the equilibrium is non negative and zero exactly at $\theta=\pi$, and the equilibria are localized when 
$\tilde K >  \tilde K(c=1) $. They converge to a Dirac mass centered at $\theta =0$ when $K \rightarrow + \infty$.
(see proposition (\ref{prop:eq_coherent_eq} ) for analytical formulas for these equilibria).
(Details on numerical methods used here can be found in the appendix.) 
}}\label{fig:bif_m15}
\end{center}
\end{figure}

\begin{figure}[h!]
\begin{center}
\psfrag{K(tilde)}[B][B][1][0]{  {$\tilde K$}}
\psfrag{x1}[B][B][1][0]{  {$x_1$}}
\psfrag{theta}[B][B][1][0]{ \small {$\theta$}}
\psfrag{u(theta)}[B][B][1][0]{ \small {$u(\theta)$}}
\psfrag{1.0/6.28}[B][B][1][0]{ \tiny{$\frac{1}{2\pi}$}}
\psfrag{12pi}[B][B][1][0]{ \tiny{$\frac{1}{2\pi}$}}
\psfrag{Kt=2.4**}[B][B][1][0]{ \tiny{$\tilde K=2.4$}}
\psfrag{Kt=2.12**}[B][B][1][0]{ \tiny{$\tilde K=2.12$}}
\psfrag{Kt=2.01**}[B][B][1][0]{ \tiny{$\tilde K=2.01$}}
\psfrag{Kt=2.5**}[B][B][1][0]{ \tiny{$\tilde K=2.5$}}
\psfrag{Kt=13.5**}[B][B][1][0]{ \tiny{$\tilde K=13.5$}}
\psfrag{Kt=6.0**}[B][B][1][0]{ \tiny{$\tilde K=6.0$}}
\psfrag{u(t,theta)}[B][B][1][0]{ \small {$u(t,\theta)$}}
\psfrag{t=0.0**}[B][B][1][0]{ \tiny{$t=0.0$}}
\psfrag{t=1.0**}[B][B][1][0]{ \tiny{$t=1.0$}}
\psfrag{t=5.0**}[B][B][1][0]{ \tiny{$t=5.0$}}
\psfrag{t=8.0**}[B][B][1][0]{ \tiny{$t=8.0$}}
\psfrag{t=10.0**}[B][B][1][0]{ \tiny{$t=10.0$}}
\psfrag{t=15.0**}[B][B][1][0]{ \tiny{$t=15.0$}}
\psfrag{m=1.5}[B][B][1][0]{ \tiny {$m=1.5$}}
\psfrag{K=1.5}[B][B][1][0]{ \tiny {$\tilde K=1.5$}}
\psfrag{K=2.12}[B][B][1][0]{ \tiny {$\tilde K=2.12$}}
\psfrag{K=6.0}[B][B][1][0]{ \tiny {$\tilde K=6.0$}}
\vspace{-3mm}
\begin{center}
\includegraphics[scale=1.0]{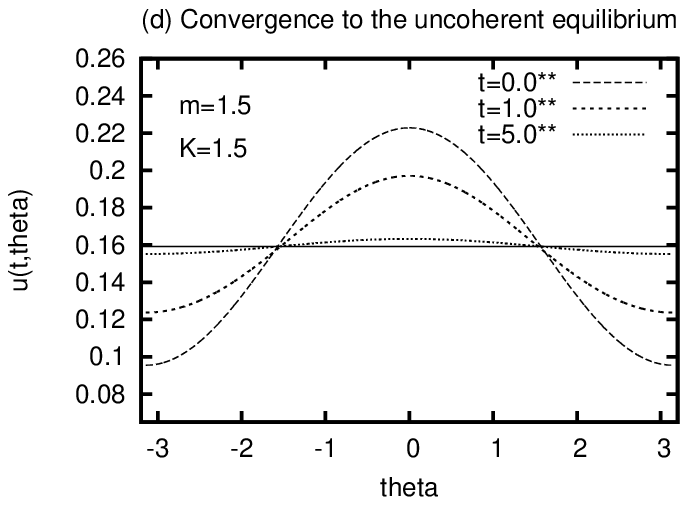}
\includegraphics[scale=1.0]{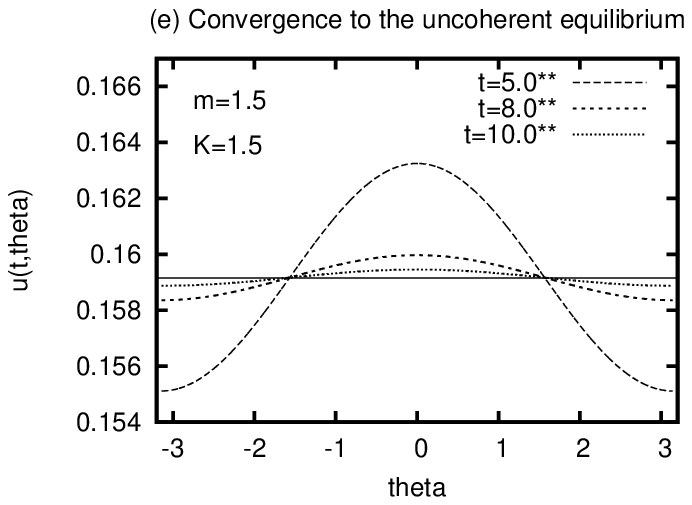}
\end{center}
\vspace{-3mm}
\begin{center}
\includegraphics[scale=1.0]{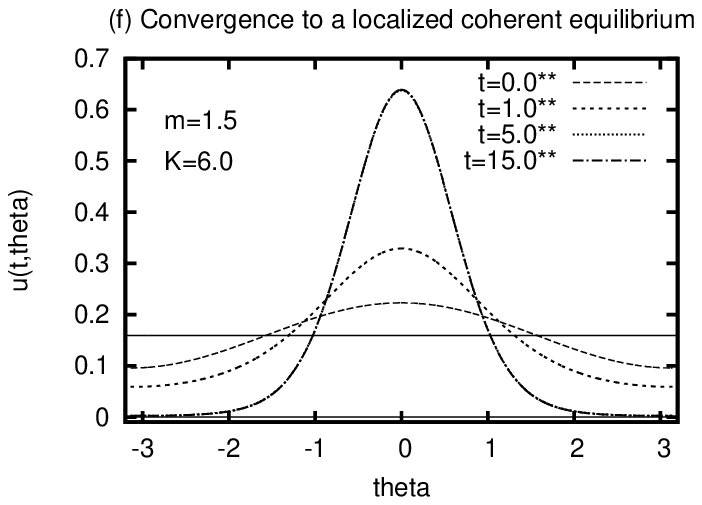}
\end{center}
\end{center}
\begin{center}
\vspace{-2mm}
\caption{Convergence of solutions in case $m=1.5$.
\small{
When $\tilde K <2$, solutions converge to the uncoherent equilibrium $(P_1)$. 
This is illustrated in panels (d) and (e) where we used for initial condition $ u_0 = \frac1{2\pi} + \frac{x_1(0)}{\pi} \cos(\cdot)$ with $x_1(0) = 0.2$, 
and we have $ \vert x_1(t) \vert < 10^{-3} $ for all $t \geq 10.0$.
Panel (f) illustrates convergence to the localized coherent equilibria $(S_l)$ for $\tilde K >  \tilde K(c=1) $.
We used $\tilde K=6.0$, the solution remains smooth and positive before converging, and
using the same initial condition that in panel (d) it has numerically converged for $t \geq 5$.
(Details on numerical methods used here can be found in the appendix.) 
}}\label{fig:cv_m15}
\end{center}
\end{figure}

\subsection{Coherence transition when $m>2$}
\label{sec:bif_m>2}

\par When $m>2$ the function $x_1 = x_1(c)$ is decreasing on $]-1, + \infty[$,
and the function $\tilde K = \tilde K(c)$ is decreasing on $]-1, c^*[$, it has a unique minimum at 
$c = c^* < 1$ and it is increasing on $]c^*, + \infty[$ (in figure \ref{fig:K_c} see the lowest curve in panel (a) for $m=3$,
and panel (b) for $m=3$, $m=5$, $m=8$, and $m=15$). 
The value $ \underset{c \in ]-1,1[}\min \tilde K_m(c)$, that depends on $m$, defines the border of regions $(U)$ and $(B)$ in figure 
\ref{fig:bif_UBC}: at this value a fold bifurcation gives birth to two (localized) coherent equilibria. 
The equilibria given by (\ref{eq:cond_u_1}) with $c<c^*$ are coherent, localized and (locally) stable. 
The equilibria given (\ref{eq:cond_u_1}) with $c<c^*$ are coherent and unstable, for $c<1$ they are localized and for $c>1$ they are positive on $\mathbb S$.
%
\par When $ \tilde K < \underset{c \in ]-1,1[}\min \tilde K_m(c)$ the unique equilibrium of (\ref{eq:main}) is the uncoherent steady state $\frac{1}{2\pi}$, 
which is locally and globally stable (label ($P_1$) figure \ref{fig:bif_m3} panel (a)). 
When $\underset{c \in ]-1,1[}\min \tilde K_m(c) < \tilde K < 2$, there is one coherent localized stable equilibrium 
denoted by label $(S_l)$ in figure \ref{fig:bif_m3} panel (a), 
and illustrated for $\tilde K=1.84$ by the lowest curve in figure \ref{fig:bif_m3} panel (d).
For the same values $\underset{c \in ]-1,1[}\min \tilde K_m(c) < \tilde K < 2$ there is also
and one coherent unstable equilibrium 
that is localized for  $\underset{c \in ]-1,1[}\min \tilde K_m(c) < \tilde K < \tilde K(c=1)$
(label $(U_l)$ in panel  figure \ref{fig:bif_m3} (a)),
and positive for  $ \tilde K(c=1) < \tilde K < 2$.
(label $(U_p)$ in  figure \ref{fig:bif_m3} panel (a)).
Panel (c) in figure \ref{fig:bif_m3} shows the shape of such localized equilibrium when $\tilde K =1.84$, 
and positive equilibria are shown for various $\tilde K(c=1) \tilde K <2$ in panel (b).
At $\tilde K= \tilde K_c = 2$ a subcritical pitchfork bifurcation occurs (see proposition \ref{prop:pitchfork_bif}).
When $\tilde K >2 $, the uncoherent steady state $\frac{1}{2\pi}$ is locally unstable (label ($P_2$) figure \ref{fig:bif_m3} panel (a))
and equation (\ref{eq:main}) has a unique localized coherent equilibrium with $x_1 >0$ 
(label $(S_l)$ in figure \ref{fig:bif_m3} panel (a)).
Two such equilibria are shown in figure \ref{fig:bif_m3} panel (d) for $\tilde K=2.4$ and $\tilde K=4.7$, 
and in panel (e) for $\tilde K=7$, $\tilde K=17$, $\tilde K=100$ and $\tilde K=400$.
\par When $\tilde K< \tilde K(c^*)$, all solutions converge to the uncoherent equilibrium $\frac{1}{2\pi}$
(region $(U)$ in figure \ref{fig:bif_UBC}, label $(P_1)$ in figure \ref{fig:bif_m3} panel (a)).
In the multistability region $(B)$ (see figure \ref{fig:bif_UBC}),
corresponding to $ \tilde K(c^*) < \tilde K< 2$, solutions typically converge either to the uncoherent equilibrium $(P_1)$ 
or to a coherent localized equilibrium $(S_l)$. 
This is illustrated in figure \ref{fig:cv_m3} panels (f) and (g). For the same value of $\tilde K=1.86$ and two different 
initial conditions, we observe convergence to a coherent equilibrium in panel (f) and convergence 
to the uncoherent equilibrium in panel (g). 
When $\tilde K >2$ (region $(C)$ in figure \ref{fig:bif_UBC}), most solutions converge to the 
localized and coherent equilibrium (label $(S_l)$ in panel (a) of figure \ref{fig:bif_m3}).
Panel (h) of figure \ref{fig:cv_m3} shows such a solution for $\tilde K = 7$.

\begin{figure}[h!tp]
\begin{center}
\psfrag{K(tilde)}[B][B][1][0]{  {$\tilde K$}}
\psfrag{x1}[B][B][1][0]{  {$x_1$}}
\psfrag{m=3}[B][B][1][0]{ \tiny {$m=3$}}
\psfrag{theta}[B][B][1][0]{ \small {$\theta$}}
\psfrag{u(theta)}[B][B][1][0]{ \small {$u(\theta)$}}
\psfrag{Kt=1.86**}[B][B][1][0]{ \tiny{$\tilde K=1.86$}}
\psfrag{Kt=1.95**}[B][B][1][0]{ \tiny{$\tilde K=1.95$}}
\psfrag{Kt=1.995**}[B][B][1][0]{ \tiny{$\tilde K=1.995$}}
\psfrag{1.0/6.28}[B][B][1][0]{ \tiny{$\frac{1}{2\pi}$}}
\psfrag{Kt=1.84**}[B][B][1][0]{ \tiny{$\tilde K=1.84$}}
\psfrag{Kt=2.4**}[B][B][1][0]{ \tiny{$\tilde K=2.4$}}
\psfrag{Kt=4.7**}[B][B][1][0]{ \tiny{$\tilde K=4.7$}}
\psfrag{Kt=7**}[B][B][1][0]{ \tiny{$\tilde K=7$}}
\psfrag{Kt=17**}[B][B][1][0]{ \tiny{$\tilde K=17$}}
\psfrag{Kt=100**}[B][B][1][0]{ \tiny{$\tilde K=100$}}
\psfrag{Kt=400**}[B][B][1][0]{ \tiny{$\tilde K=400$}}
\vspace{-1.0cm}
\begin{center}
\includegraphics[scale=0.8]{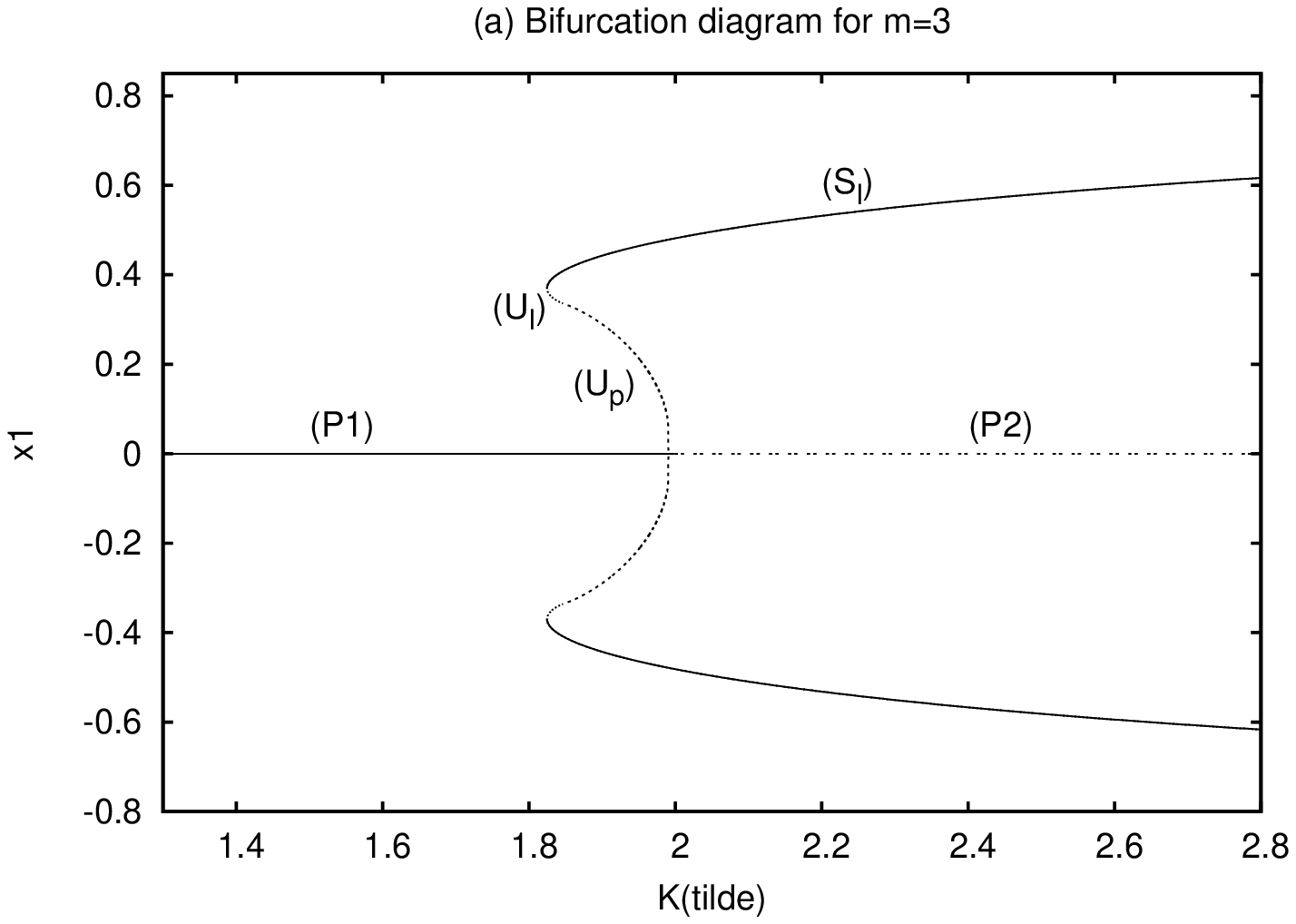}
\end{center}
\vspace{-6mm}
\begin{center}
\includegraphics[scale=0.9]{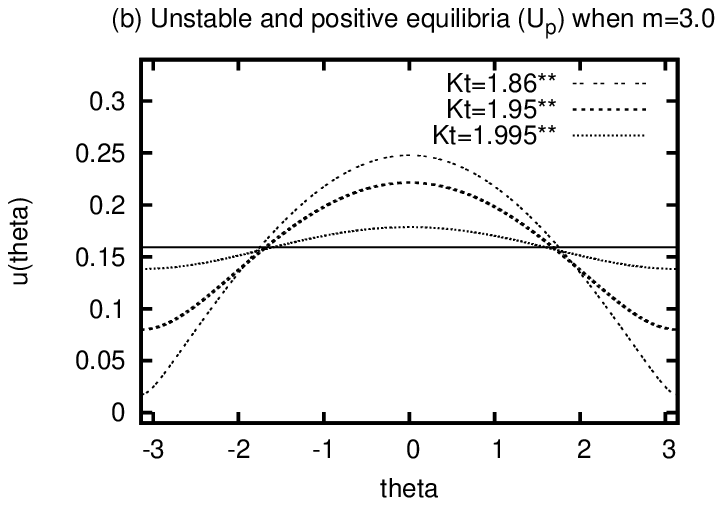}
\includegraphics[scale=0.9]{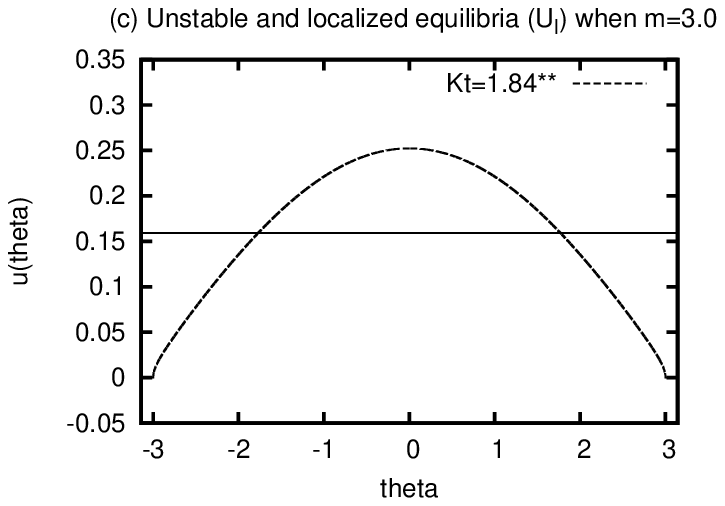}
\end{center}
\vspace{-4mm}
\begin{center}
\includegraphics[scale=0.9]{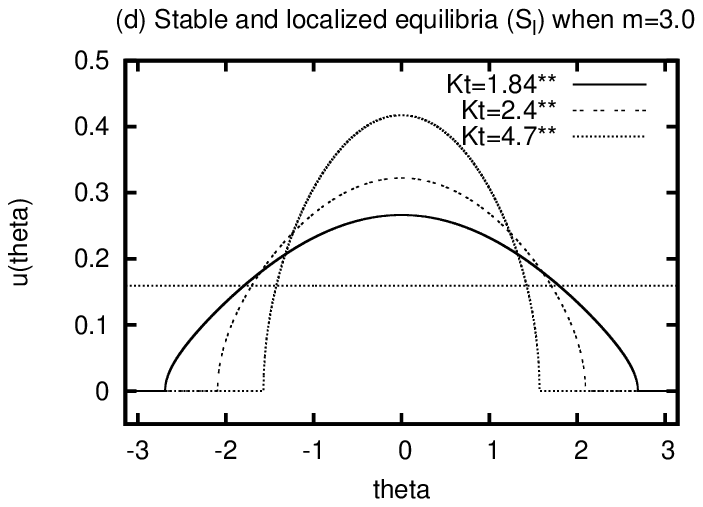}
\includegraphics[scale=0.9]{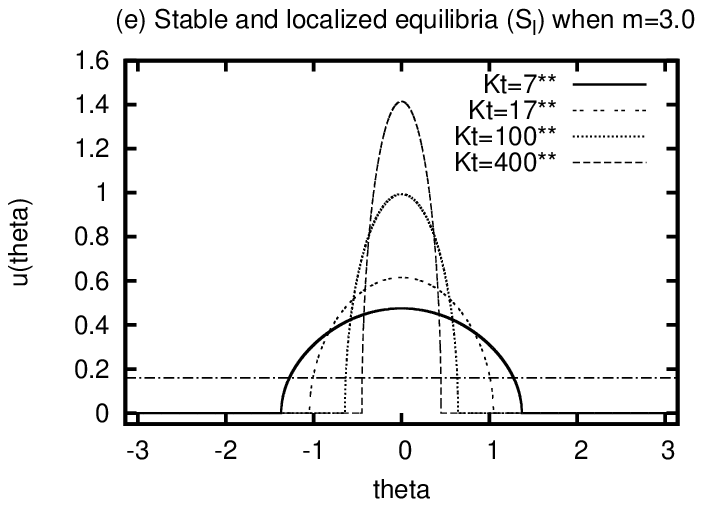}
\end{center}
\end{center}
\vspace{-3mm}
\begin{center}
\caption{Equilibria bifurcation diagram, case $m=3$. 
\small{
Panel (a) shows the number of equilibria of (\ref{eq:main}) for $0 < \tilde  K < + \infty$.
We see in particular the subcritical pitchfork bifurcation occurring at $\frac1{2\pi}$ when $\tilde K = 2$. 
The branch of unstable equilibria that disappear satisfies 
$ \tilde K - \tilde K_c \sim -(m-2) x_1^2 > 0 $ in a neighborhood of the bifurcation (see proposition \ref{prop:pitchfork_bif}).
In panels (b) to (e), the horizontal line represents $\frac1{2\pi}$. 
Panel (b) shows the shape of unstable coherent equilibria $(U_p)$, for various values of $\tilde K$ in 
$] \tilde K(c=1), 2  [$.
These equilibria are non negative when $\tilde K=\tilde K(c=1)$, and positive when $] \tilde K(c=1), 2  [$. 
They converge to $\frac1{2\pi}$ when $\tilde K \rightarrow 2$.
The shape of localized and unstable coherent equilibria $(U_l)$, that exist for $\tilde K$ in 
$]\min \tilde K_m(c), \tilde K(c=1)  [$ is illustrated in panel (c) for $\tilde K=1.84$. 
These equilibria appear by fold bifurcation at $\tilde K = \min \tilde K_m(c)$, 
and they become positive when $\tilde K > \tilde K(c=1) $.
Panels (d) and (e) show localized coherent (locally) stable equilibria $(S_l)$, for various values of $\tilde K > \min \tilde K_m(c)$ 
(see proposition (\ref{prop:eq_coherent_eq} ) for analytical formulas for coherent equilibria).
These equilibria appear at $\tilde K = \min \tilde K_m(c)$  by fold bifurcation,
 they are localized and stable for all  $\tilde K > \min \tilde K_m(c)$,
and they converge to a Dirac mass at $\theta = 0$ when $\tilde K \rightarrow + \infty$. 
(Details on numerical methods used here can be found in the appendix.)
}}
\label{fig:bif_m3}
\end{center}
\end{figure}

\begin{figure}[h!tp]

\begin{center}
\psfrag{K(tilde)}[B][B][1][0]{  {$\tilde K$}}
\psfrag{x1}[B][B][1][0]{  {$x_1$}}
\psfrag{m=2}[B][B][1][0]{ \tiny {$m=2$}}
\psfrag{Kt=2}[B][B][1][0]{ \small {$\tilde K=2$}}
\psfrag{theta}[B][B][1][0]{ \small {$\theta$}}
\psfrag{u(theta)}[B][B][1][0]{ \small {$u(\theta)$}}
\psfrag{1.0/6.28}[B][B][1][0]{ \tiny{$\frac{1}{2\pi}$}}
\psfrag{u(t,theta)}[B][B][1][0]{ \small {$u(t,\theta)$}}
\psfrag{t=0.0**}[B][B][1][0]{ \tiny{$t=0.0$}}
\psfrag{t=1.0**}[B][B][1][0]{ \tiny{$t=1.0$}}
\psfrag{t=5.0**}[B][B][1][0]{ \tiny{$t=5.0$}}
\psfrag{t=8.0**}[B][B][1][0]{ \tiny{$t=8.0$}}
\psfrag{t=10.0**}[B][B][1][0]{ \tiny{$t=10.0$}}
\psfrag{t=15.0**}[B][B][1][0]{ \tiny{$t=15.0$}}
\psfrag{t=30.0**}[B][B][1][0]{ \tiny{$t=30.0$}}
\psfrag{t=100.0**}[B][B][1][0]{ \tiny{$t=100.0$}}
\psfrag{t=300.0**}[B][B][1][0]{ \tiny{$t=300.0$}}
\psfrag{m=3.0}[B][B][1][0]{ \tiny {$m=3.0$}}
\psfrag{K=1.86}[B][B][1][0]{ \tiny {$\tilde K=1.86$}}
\psfrag{K=2.0}[B][B][1][0]{ \tiny {$\tilde K=2.0$}}
\psfrag{K=7.0}[B][B][1][0]{ \tiny {$\tilde K=7.0$}}
\begin{center}
\includegraphics[scale=0.9]{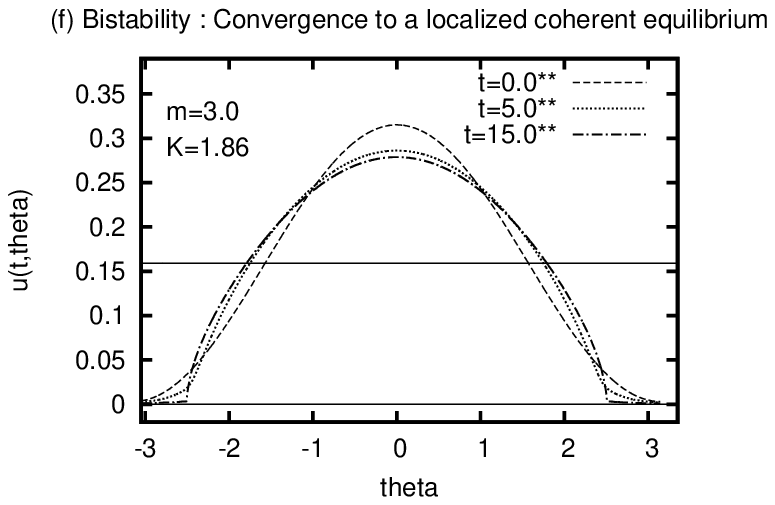}
\includegraphics[scale=0.9]{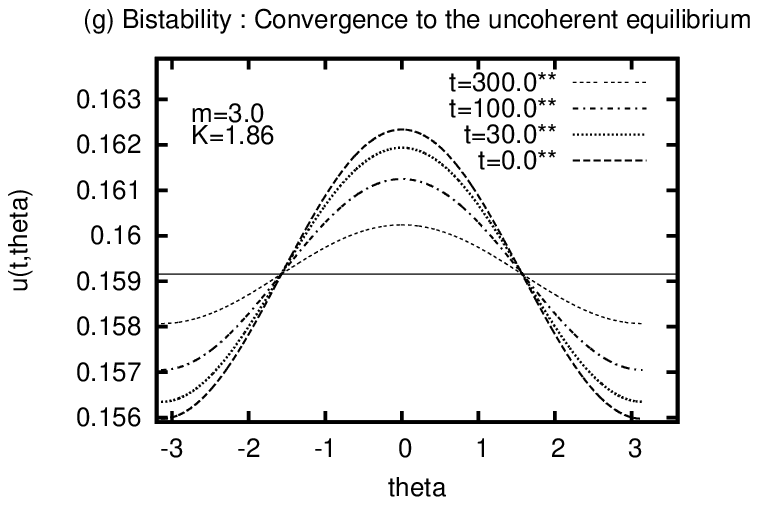}
\end{center}
\vspace{-3mm}
\begin{center}
\includegraphics[scale=0.9]{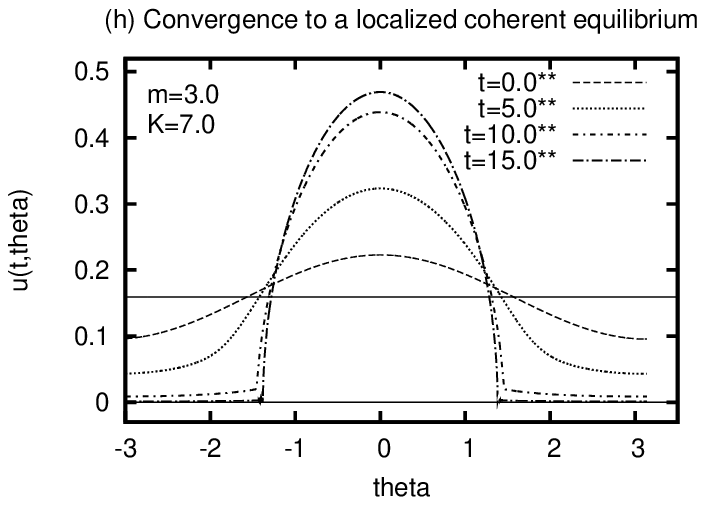}
\end{center}
\end{center}
\begin{center}
\vspace{-2mm}
\caption{Convergence of solutions to equilibria, case $m=3$. 
\small{
When $\tilde K < \min \tilde K_m(c)$, solutions converge to the uncoherent equilibrium (region $(U)$ in figure \ref{fig:bif_UBC}). 
Panels (f) and (g) illustrate the bistability phenomena in region $(B)$ of figure \ref{fig:bif_UBC} 
(for $\tilde K$ in  $]\min \tilde K_m(c), 2 [$).
Using $ u_0 = \frac1{2\pi} + \frac{x_1}{\pi} \cos(\cdot)$ with $x_1=0.5$, 
panel (f) shows convergence of one solution to the localized coherent equilibria $(S_l)$ when $K=1.86$,
and for the same value of $\tilde K$ panel (g) shows that 
with different initial data ($x_1=0.01$) the solution converges to the uncoherent equilibria. 
When $\tilde K >2$ solutions remain positive and converge to the localized coherent equilibria $(S_l)$ 
(region $(C_l)$ in figure \ref{fig:bif_UBC}), this is illustrated in panel (h) for $\tilde K=7$. 
In panels (f) and (h) the solutions have numerically converged at $t=15.0$ and remain the same for all $t \geq 15.0$. 
In panel (g) convergence is much slower, using $x_1=0.01$ for the initial condition, the solution is still approaching 
$\frac 1 {2\pi}$ when $t=300.0$.
(Details on numerical methods used here can be found in the appendix.)
}
}\label{fig:cv_m3}
\end{center}
\end{figure}

\subsection{Coherence transition when $m=2$}
\label{sec:bif_m=2}

\par When $m=2$ the function $\tilde K = \tilde K(c)$ is decreasing on $]-1, 1[$ and constant on $[1, +\infty[$ 
(proposition \ref{prop:eq_coherent_eq_m2}, see illustations in figure \ref{fig:K_c} panel (a)), and
the function $x_1 = x_1(c)$ is decreasing on $]-1, + \infty[$. 
The bifurcation diagram $x_1 = x_1(\tilde K)$ is shown in figure \ref{fig:bif_m2} panel (a).

\par When $ \tilde K < \tilde K_c=2$ the unique equilibrium of (\ref{eq:main}) is the uncoherent steady state $\frac{1}{2\pi}$ 
(corresponding to $x_1=0$), which is locally and globally stable (label ($P_1$) figure \ref{fig:bif_m2} panel (a)). 
At $\tilde K= \tilde K_c = 2$ a degenerate pitchfork bifurcation occurs (proposition \ref{prop:eq_coherent_eq_m2}). 
Any $u$ given by $u(\cdot) = \frac1{2\pi} + \frac1\pi x_1 \cos(\cdot)$ with $0 \leq x_1 \leq \frac 12$ is an equilibrium of (\ref{eq:main}) 
(label ($N_p$) figure \ref{fig:bif_m15} panel (a)). The uncoherent equilibrium is recovered when $x_1 = 0$, for any $ 0 < x_1 < 0.5$ the equilibrium is 
positive on $[-\pi, \pi]$, and for $x_1 = 0.5$ the equilibrium is non negative on $[-\pi, \pi]$. 
When $m=2$ the transition from uncoherence to coherence is sharp. The family of equilibria connecting $\frac1{2\pi}$ to the localized coherent equilibria 
appear all-at-once when $\tilde K=2$. In particular we have $\tilde K - \tilde K_c = o( x_1^n )$ for any $n \geq 2$ in some neighborhood of the bifurcation point.
As soon as $\tilde K >2 $, the uncoherent steady state $\frac{1}{2\pi}$ is locally unstable (label ($P_2$) figure \ref{fig:bif_m15} panel (a))
and equation (\ref{eq:main}) has a unique localized coherent equilibrium with $x_1 > \frac 12$ (label ($S_l$) figure \ref{fig:bif_m2} panel (a)). 

\par 
When $\tilde K< \tilde K_c$, all solutions converge to the uncoherent equilibrium $\frac{1}{2\pi}$.
When $\tilde  K=2$ the phase space is foliated by the stable manifolds of the equilibria 
$u(\cdot) = \frac1{2\pi} + \frac1\pi x_1 \cos(\cdot)$ where $0 \leq x_1 \leq \frac 12$,
that is to say: any solution converge to one of those equilibria, which depends on the initial data $u_0$.
Panels (d) and (e) in figure \ref{fig:cv_m2} show two solutions of (\ref{eq:main}) for the same value of $\tilde K=2$ and two different initial conditions 
that converge to two different equilibria of the family $(N_p)$.
When $\tilde K > 2$, solutions typically converge to the unique (localized) coherent equilibrium $(S_l)$ of (\ref{eq:main}).
In figure \ref{fig:cv_m2}, panel (f) show this convergence phenomena for one solution when $\tilde K =6.0$.

\par Bifurcation diagrams for $m<2$, $m=2$ and $m>2$ are compared in figure \ref{fig:bif_m>=<2}.
The value of $\tilde K$ at which stable and localized equilibria $(S_l)$ appear depends smoothly on $m$, 
it equals $\underset{c}{\min} \tilde K_m(c)$ when $m< 2$, $2$ when $m=2$ and $\tilde K_m(c=1)$ when $m>2$.
The unstable equilibria $(U_p)$ and $(U_l)$ that exist for $m<2$ converge to the family of equilibria $(N_p)$ at $m=2$ and $\tilde K=2$. 
Similarly the stable and positive equilibria $(S_p)$ that exist when $m>2$ converge to the family of equilibria $(N_p)$ at $m=2$ and $\tilde K=2$.

\begin{figure}[h!tp]
\begin{center}
\psfrag{K(tilde)}[B][B][1][0]{  {$\tilde K$}}
\psfrag{x1}[B][B][1][0]{  {$x_1$}}
\psfrag{m=2}[B][B][1][0]{ \tiny {$m=2$}}
\psfrag{Kt=2}[B][B][1][0]{ \small {$\tilde K=2$}}
\psfrag{theta}[B][B][1][0]{ \small {$\theta$}}
\psfrag{u(theta)}[B][B][1][0]{ \small {$u(\theta)$}}
\psfrag{1.0/6.28}[B][B][1][0]{ \tiny{$\frac{1}{2\pi}$}}
\psfrag{Kt=2.5**}[B][B][1][0]{ \tiny{$\tilde K=2.5$}}
\psfrag{Kt=6.0**}[B][B][1][0]{ \tiny{$\tilde K=6.0$}}
\psfrag{Kt=13.5**}[B][B][1][0]{ \tiny{$\tilde K=13.5$}}
\psfrag{c=1.0**}[B][B][1][0]{ \tiny{$c=1.0$}}
\psfrag{c=2.0**}[B][B][1][0]{ \tiny{$ c=2.0$}}
\psfrag{c=5.0**}[B][B][1][0]{ \tiny{$ c=5.0$}}
\psfrag{c=9.0**}[B][B][1][0]{ \tiny{$ c=9.0$}}
\psfrag{u(t,theta)}[B][B][1][0]{ \small {$u(t,\theta)$}}
\psfrag{t=0.0**}[B][B][1][0]{ \tiny{$t=0.0$}}
\psfrag{t=1.0**}[B][B][1][0]{ \tiny{$t=1.0$}}
\psfrag{t=5.0**}[B][B][1][0]{ \tiny{$t=5.0$}}
\psfrag{t=8.0**}[B][B][1][0]{ \tiny{$t=8.0$}}
\psfrag{t=10.0**}[B][B][1][0]{ \tiny{$t=10.0$}}
\psfrag{t=15.0**}[B][B][1][0]{ \tiny{$t=15.0$}}
\psfrag{m=2.0}[B][B][1][0]{ \tiny {$m=2.0$}}
\psfrag{K=1.5}[B][B][1][0]{ \tiny {$\tilde K=1.5$}}
\psfrag{K=2.0}[B][B][1][0]{ \tiny {$\tilde K=2.0$}}
\psfrag{K=6.0}[B][B][1][0]{ \tiny {$\tilde K=6.0$}}
\begin{center}
\includegraphics[scale=0.7]{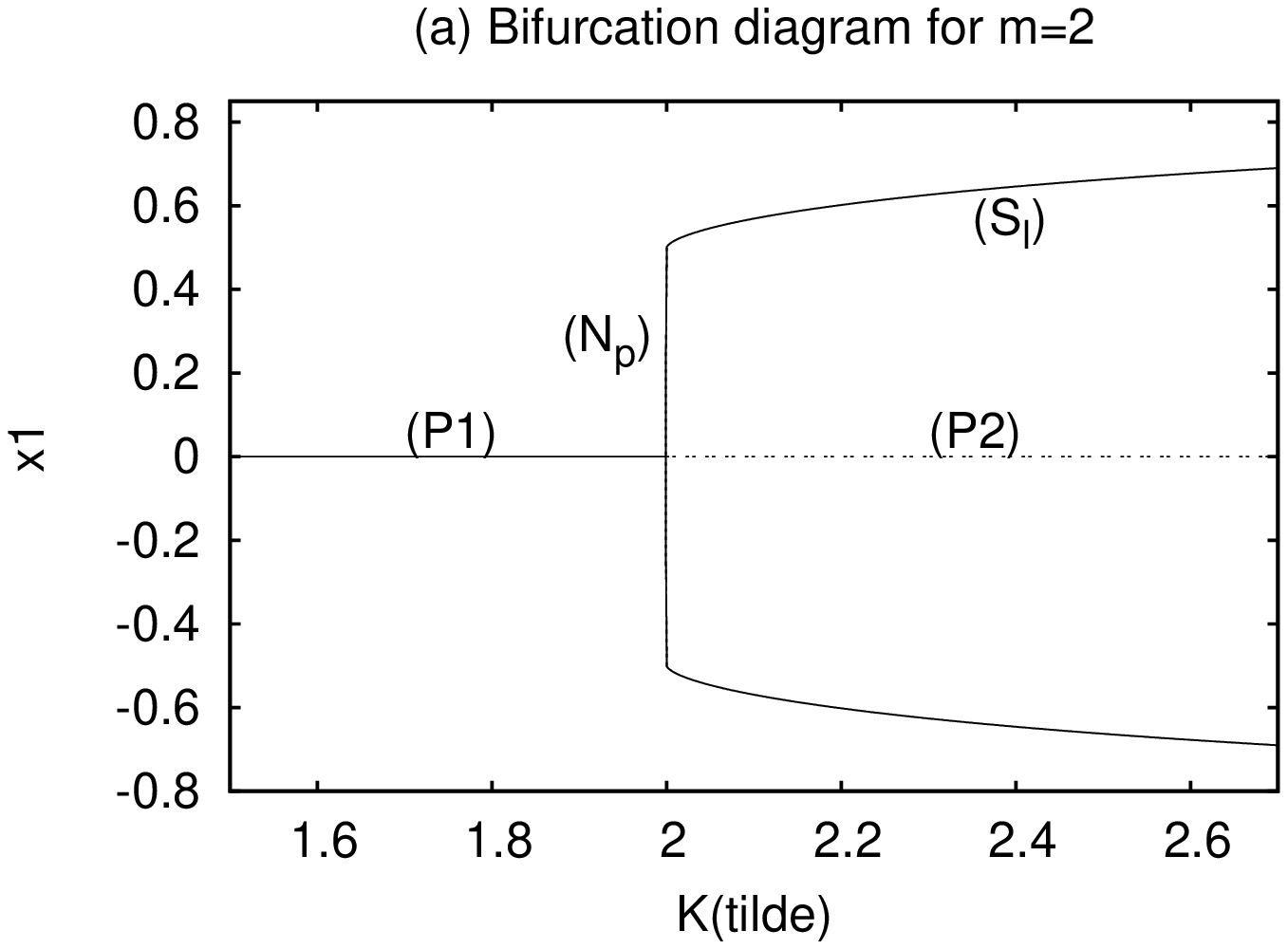}
\end{center}
\vspace{-4mm}
\begin{center}
\includegraphics[scale=0.9]{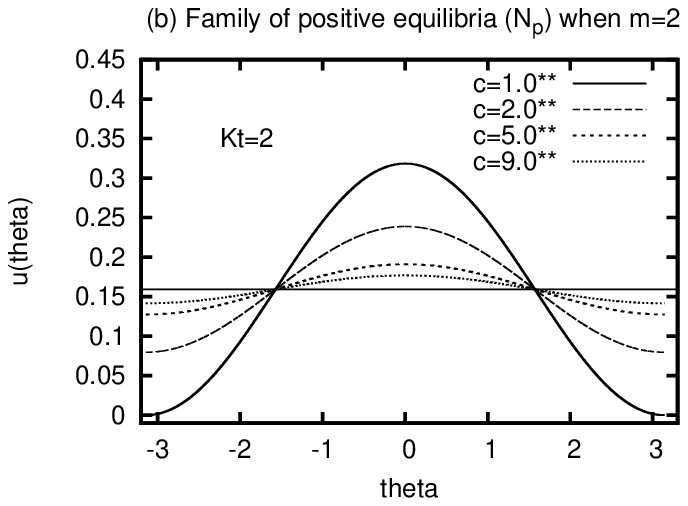}
\includegraphics[scale=0.9]{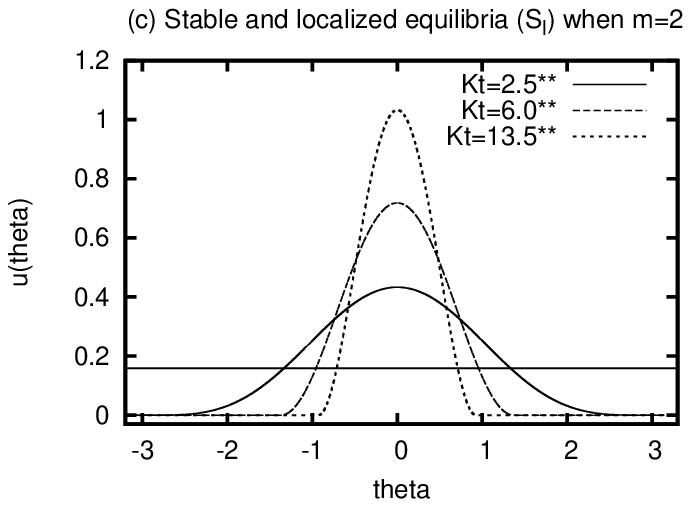}
\end{center}
\end{center}
\begin{center}
\vspace{-2mm}
\caption{Equilibria bifurcation diagram, equilibria of (\ref{eq:main}) and convergence of solutions in case $m=2$. 
\small{
Panel (a) show the number of equilibria of (\ref{eq:main}) for $0 < \tilde  K < + \infty$.
We see in particular the degenerate pitchfork bifurcation occurring at $\frac1{2\pi}$ when $\tilde K = 2$. 
The branch of equilibria that appear at $\tilde K_c$ is flat:  there is a 1-parameter family of equilibria 
$u(\cdot) = \frac1{2\pi} + \frac1\pi x_1 \cos(\cdot)$ where $0 \leq x_1 \leq \frac 12$, connecting $\frac{1}{2\pi}$ (for $x_1 = 0$)
to the onset of localized equilibria (for $x_1 = \frac 12$) (see proposition \ref{prop:eq_coherent_eq_m2}).
Positive coherent equilibria of his family $(N_p)$ are shown in panel (b), for $\tilde K=2$ and various values of $x_1 = x_1(c)$, 
and confirm that the equilibrium $u$ converge to $\frac1{2\pi}$ has $c \rightarrow + \infty$ and $x_1 \rightarrow 0$.
Panel (c) shows the shape of coherent equilibria $(S_l)$, for various different values of $\tilde K > \tilde K(c=1)$ 
(see proposition (\ref{prop:eq_coherent_eq} ) for analytical formulas for these equilibria). 
These equilibria are smooth and localized for all $K > 2$, and they converge to a Dirac mass at $\theta=0$ when $K \rightarrow + \infty$.
(In panels (b) and (c), the horizontal line represents $\frac1{2\pi}$, details on numerical methods used here can be found in the appendix.) 
}}\label{fig:bif_m2}
\end{center}
\end{figure}

\begin{figure}[h!tp]
\begin{center}
\psfrag{K(tilde)}[B][B][1][0]{  {$\tilde K$}}
\psfrag{x1}[B][B][1][0]{  {$x_1$}}
\psfrag{m=2}[B][B][1][0]{ \tiny {$m=2$}}
\psfrag{Kt=2}[B][B][1][0]{ \small {$\tilde K=2$}}
\psfrag{theta}[B][B][1][0]{ \small {$\theta$}}
\psfrag{u(theta)}[B][B][1][0]{ \small {$u(\theta)$}}
\psfrag{1.0/6.28}[B][B][1][0]{ \tiny{$\frac{1}{2\pi}$}}
\psfrag{Kt=2.5**}[B][B][1][0]{ \tiny{$\tilde K=2.5$}}
\psfrag{Kt=6.0**}[B][B][1][0]{ \tiny{$\tilde K=6.0$}}
\psfrag{Kt=13.5**}[B][B][1][0]{ \tiny{$\tilde K=13.5$}}
\psfrag{c=1.0**}[B][B][1][0]{ \tiny{$c=1.0$}}
\psfrag{c=2.0**}[B][B][1][0]{ \tiny{$ c=2.0$}}
\psfrag{c=5.0**}[B][B][1][0]{ \tiny{$ c=5.0$}}
\psfrag{c=9.0**}[B][B][1][0]{ \tiny{$ c=9.0$}}
\psfrag{u(t,theta)}[B][B][1][0]{ \small {$u(t,\theta)$}}
\psfrag{t=0.0**}[B][B][1][0]{ \tiny{$t=0.0$}}
\psfrag{t=1.0**}[B][B][1][0]{ \tiny{$t=1.0$}}
\psfrag{t=3.0**}[B][B][1][0]{ \tiny{$t=3.0$}}
\psfrag{t=5.0**}[B][B][1][0]{ \tiny{$t=5.0$}}
\psfrag{t=8.0**}[B][B][1][0]{ \tiny{$t=8.0$}}
\psfrag{t=10.0**}[B][B][1][0]{ \tiny{$t=10.0$}}
\psfrag{t=11.0**}[B][B][1][0]{ \tiny{$t=11.0$}}
\psfrag{t=15.0**}[B][B][1][0]{ \tiny{$t=15.0$}}
\psfrag{m=2.0}[B][B][1][0]{ \tiny {$m=2.0$}}
\psfrag{K=1.5}[B][B][1][0]{ \tiny {$\tilde K=1.5$}}
\psfrag{K=2.0}[B][B][1][0]{ \tiny {$\tilde K=2.0$}}
\psfrag{K=6.0}[B][B][1][0]{ \tiny {$\tilde K=6.0$}}
\begin{center}
\includegraphics[scale=1.0]{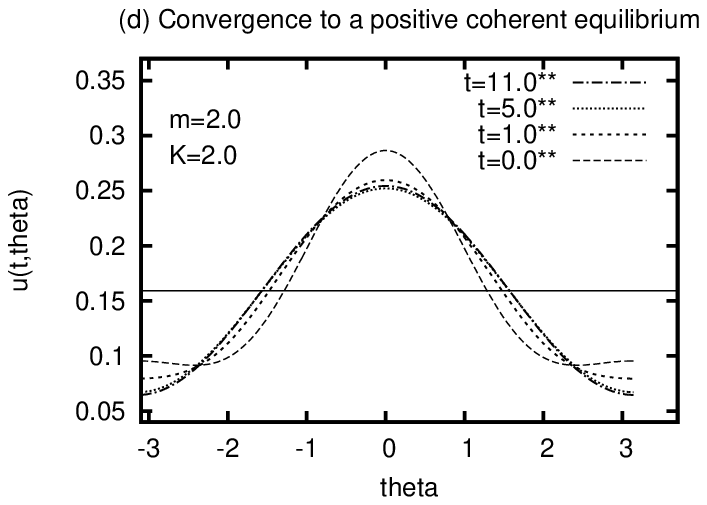}
\includegraphics[scale=1.0]{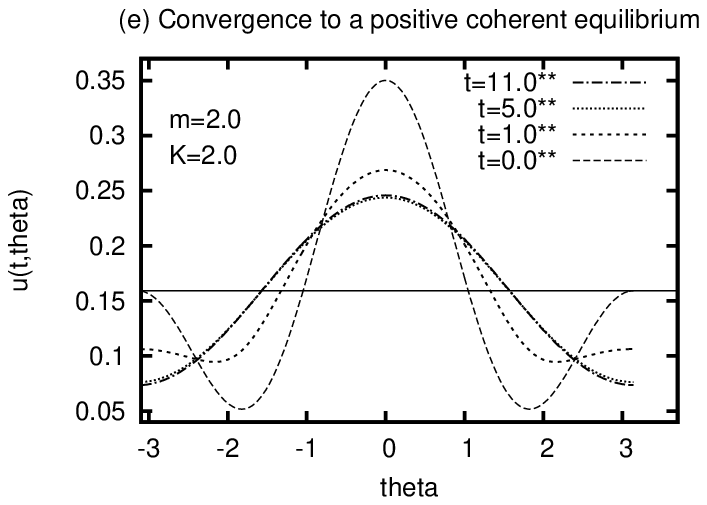}
\end{center}
\vspace{-3mm}
\begin{center}
\includegraphics[scale=1.0]{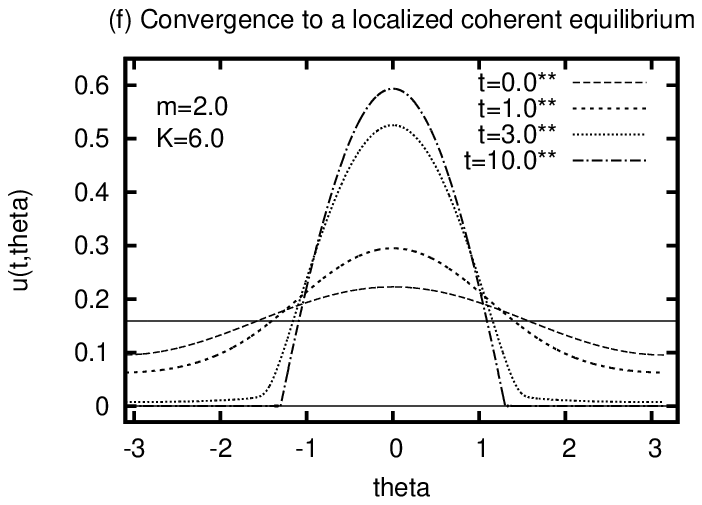}
\end{center}
\end{center}
\begin{center}
\vspace{-2mm}
\caption{Convergence of solutions in case $m=2$. 
\small{
Panels (d) and (e) show convergence of solutions to different coherent equilibria when $K=2$. We used 
$u_0(\theta) = \frac1{2\pi} + \frac{a}\pi \cos( \theta) + \frac{b}\pi \cos( \theta) $ with $a=0.3$, $b=0.1$ in (d), 
and $a=0.4$, $b=0.2$ in (e). The solutions converge to $u(\cdot) = \frac1{2\pi} + \frac1\pi x_1 \cos(\cdot)$ with 
$x_1 \approx 0.3$ in (d) and $x_1 \approx 0.36$ in (e).
When $\tilde K >2$ solutions remains positive and converge to the localized coherent equilibria $(S_l)$ of (\ref{eq:main}).
This is illustrated in panel (f) where the solution has numerically converged at $t=5.0$ and remains the same for all $t \geq 5.0$. 
(In panels (d) to (f), the horizontal line represents $\frac1{2\pi}$, details on numerical methods used here can be found in the appendix.) 
}}\label{fig:cv_m2}
\end{center}
\end{figure}

\begin{figure}[h!tp]
\psfrag{K(tilde)}[B][B][1][0]{ \large {$\tilde K$}}
\psfrag{x1}[B][B][1][0]{ \large {$x_1$}}
\begin{center}
\includegraphics[scale=0.7]{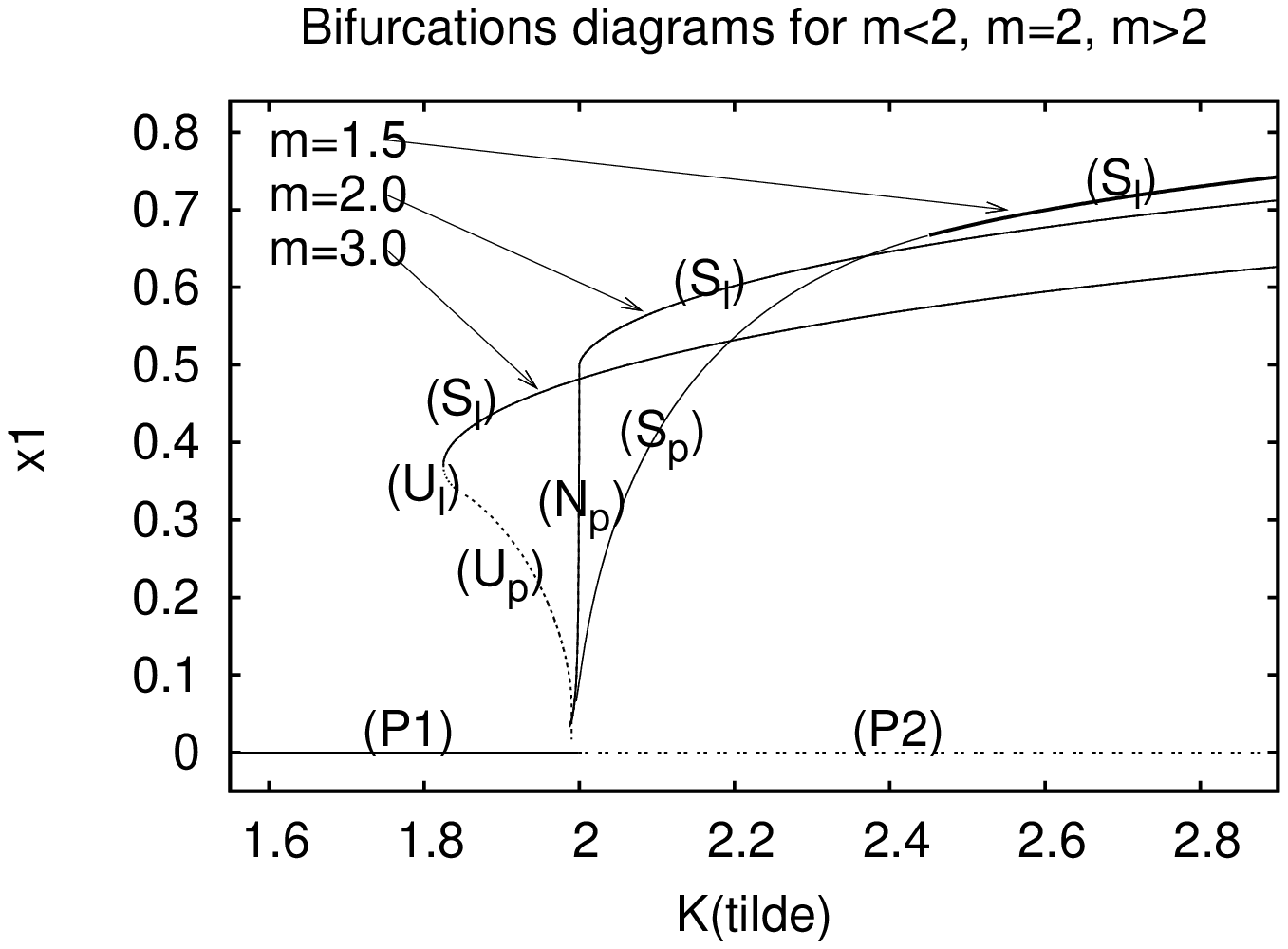}
\end{center}
\begin{center}
\vspace{-2mm}
\caption{Transition in bifurcations diagrams at $m=2$. 
\small{The value of $\tilde K$ at which stable and localized equilibria $(S_l)$ appear depends smoothly on $m$, 
it equals $\underset{c}{\min} \tilde K_m(c)$ when $m< 2$, it equals $2$ when $m=2$ and $\tilde K_m(c=1)$ when $m>2$.
The unstable equilibria $(U_p)$ and $(U_l)$ that exist for $m<2$ converge to the family of equilibria $(N_p)$ at $m=2$ and $\tilde K=2$. 
Similarly the stable and positive equilibria $(S_p)$ that exist when $m>2$ converge to the family of equilibria $(N_p)$ at $m=2$ and $\tilde K=2$.}
}\label{fig:bif_m>=<2}
\end{center}
\end{figure}

\subsection{Coherence transition when $m \rightarrow + \infty$}
\label{sec:bif_m_grand}

\par The limit $m \rightarrow + \infty$ is a limit of slow diffusion in (\ref{eq:main}).
The bifurcation scenari in the limit $m$ large is illustrated in figure \ref{fig:bif_m_grand}. 
The fold bifurcation point $\tilde K = \min_{c} \tilde K(c)$ coonverge to the pitchfork bifurcation point 
$\tilde K =2$,  
and
the interval $\tilde K \in [\min_c \tilde K(c), 2[$ in which equation (\ref{eq:main}) is bistable vanishes in the large $m$ limit, 
while the coherent equilibria branch $(S_l)$ converges pointwisely to the unstable uncoherent equilibrium branch $(P_2)$
(see illustrations of these convergence phenomena in figure \ref{fig:bif_m_grand}). 


\begin{figure}[h!tp]
\psfrag{K(tilde)}[B][B][1][0]{ \large {$\tilde K$}}
\psfrag{x1}[B][B][1][0]{ \large {$x_1$}}
\begin{center}
\includegraphics[scale=0.7]{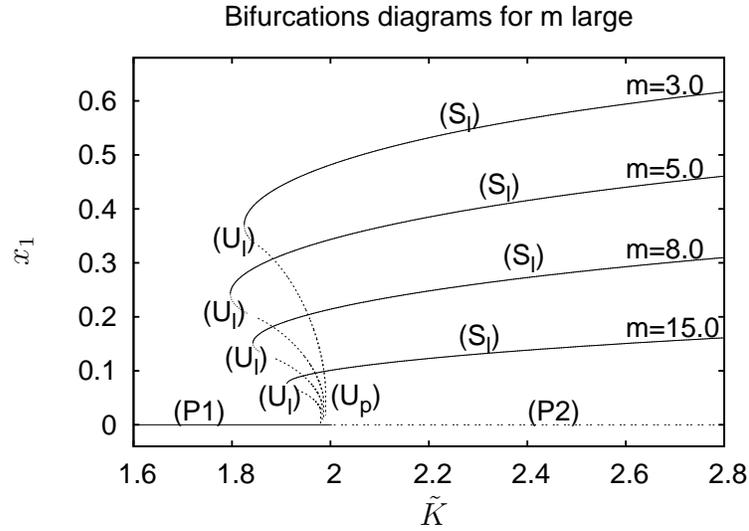}
\end{center}
\begin{center}
\vspace{-2mm}
\caption{Bifurcations diagrams for large $m$. 
\small{ In the large $m$ limit, the fold bifurcation value $\tilde K = \min_c \tilde K(c)$ converges to $\tilde K=2$,
 and the uncoherent equilibria $(U_l)$, $(U_p)$ exist in a vanishing intervall of $\tilde K$
(corresponding to the bistability region $(B)$ is figure \ref{fig:bif_UBC}).  
The coherent localized stable equilibria branch $(S_l)$ converge pointwisely to the uncoherent equilibria branch ($P_2$).  }
}
\end{center}
\label{fig:bif_m_grand}
\end{figure}

\section{Discussion}
\label{sec:discussion}

In this work, we have shown that nonlinear diffusion alters the
dynamics of eq (\ref{eq:main}) in two key aspects. One relates to the
very shape of coherent equilibria that become localized when
advection is strong. 
The other modification is that for $m>2$, the transition to
coherence scenario is through a subcritical rather than a
supercritical pitchfork bifurcation, thereby leading to existence of a
multistable regime that has no counter part in the case of linear
diffusion. Central to this impact of nonlinear diffusion is the fact
that the pitchfork bifurcation for $m=2$ is highly degenerate and acts
as an organizing center. In the following paragraphs we first discuss
the implications of this organizing center and then the generality of
our results.

Schematically, organizing centers are degenerate bifurcation points
where a number of qualitatively distinct regimes come to meet. For eq
(\ref{eq:main})  the highly degenerate pitchfork bifurcation taking
place at $(\tilde K, m)=(2,2)$ plays this pivotal role. We claim that
small perturbations of  eq (\ref{eq:main}) at this point can produce
arbitrary one-dimensional dynamics of the interval.  Indeed, a full
continuous interval of equilibria exists at this point. This fact
combined with the claim (based upon numerical investigations) that the
linearized operator at these equilibria possesses a spectral gap
suggest that this equilibria interval is a normally hyperbolic
invariant manifold. This property in turn implies that this manifold
would persist under small perturbations of eq (\ref{eq:main}).
Finally, the fact that dynamics on the equilibria set is trivial
allows one to construct ad-hoc small perturbations to reproduce
arbitrary dynamics.

A related construction exists for linear diffusion where it has been
proved that small perturbations can produce arbitrary dynamics of the
circle \cite{GPPP}. While the main arguments, i.e. the existence of a
normally hyperbolic contiuum of equilibria are similar in both cases,
there exists a main difference. For linear diffusion, the continuum of
equilibria does not exist if one considers even solutions, whereas for
$m=2$, it does so. In other words, if we relax the constraint on
solutions beeing even, we can expect small perturbations of  eq
(\ref{eq:main})  to produce arbitrary dynamis on the disk (and not
only the interval or the circle). In this sense, our results suggest
that the dynamics  with nonlinear diffusion are richer than those with
linear diffusion and give a precise meaning to this property.

In our work, results were presented for the case where the nonlocal
advection term takes on the form of a convolution with a sine
function. Here, we discuss extentions of our results to other forms of
advection.

\par For the sake of simplicity, we have restricted our analysis to the coupling term $J\ast u = \int_\mathbb S \sin(\cdot - \varphi) u(\varphi) d\varphi$. 
Results about convergence of equilibria can be extended to more general coupling terms, including the Maier-Saupe potential
$J\ast u = \int_\mathbb S \sin 2(\cdot - \varphi) u(\varphi) d\varphi$ for example
and more general non local convolution with an odd function. In the case of Maier Saupe potential 
the bifurcation scenario remains very similar to figure \ref{fig:bif_UBC}. The main difference is that coherent equilibria
 have two maximum on $[-\pi, \pi]$, 
and localized coherent equilibria have non-connected support of the form $]\pi - \epsilon, \pi + \epsilon[ \cup ]-\epsilon, \epsilon[$. 

\par Polynomial advecting terms have been considered in the context of Keller-Segel model 
with degenerate diffusion for example see \cite{MR2263432,2009arXiv0902.1878S,MR2501355} .
The corresponding equations 
\begin{equation}\label{eq:main_q}
\left \{
\begin{array}{rll}
\partial_t u &= \partial_\theta^2 (u^m) + \partial_\theta \left( u^q J \ast u \right), \;\; &  t> 0, \, \theta \in [0,2\pi], \\
u(0,\theta) &= u_0(\theta) \;\; & \theta \in [0,2\pi],  \\
u(t, 0) &= u(t, 2 \pi) \;\; & t\geq 0, \\
\partial_\theta u(t, 0) \;\; &= \partial_\theta u(t, 2 \pi) &  t \geq 0, 
\end{array} \right .
\end{equation}
with $1 \leq q < m $ can be treated as the case $q=1$ and the bifurcation scenari are the same, 
with value $m$ changed to $m-(q-1)$.
Here we simply show that proposition \ref{prop:eq_coherent_eq} on the equilibria set of (\ref{eq:main})
 can be directly extended to the following corollary 
\begin{corollaire*}
Let $u$ be a stationary solution of (\ref{eq:main_q}) with $1 \leq q < m $.
Then we have 
$$ u^{m-q}(\theta) = \frac{m-q}{m} K x_1 \left[ \cos(\theta) +c \right]_+ $$
where $x_1$ and $K$ are given by 
$$ x_1=\frac{J_{m-(q-1)}}{I_{m-(q-1)}}, \; K= \frac{m}{m-q} \frac1{x_1} \frac1{I_{m-(q-1)}^{m-q}} 
=  \frac{m}{m-q} \frac1{J_{m-(q-1)}} \frac1{I_{m-(q-1)}^{m-q-1}}  $$
for $c \in ]-1, +\infty[$, where $I_m$ and $J_m$ are defined as in proposition \ref{prop:eq_coherent_eq}.
\end{corollaire*}

\section*{Acknowledgments}
Xavier Pellegrin would like to thank Michael Goldman for interesting discussions on gradient flows. 

\section*{Appendix}

\par Figure \ref{fig:bif_m15} panels (a) to (c), figure \ref{fig:bif_m2} panels (a) to (c), and figure \ref{fig:bif_m3} panels (a) to (e),
were made using formulas (\ref{eq:K_x_1}) and (\ref{eq:cond_u_1}) that give $x_1$, $\tilde K$ and the coherent equilibria
 $u$ as functions of $c$.
The functions $I_m(c)$ and $J_m(c)$ were approximated by the classical trapezoidal rule with space discretization $\Delta_\theta=10^{-3}$
for $c \in ]-1, 2[$ with discretization $\Delta_c=10^{-2}$.
The coherent equilibria are stable exactly when $c<1$ and $\tilde K > \min \tilde K(c)$. The coherent equilibria are localized exactly when 
$c<1$. The bifurcation diagram in a neighborhood of $ \tilde K=2$ and $x_1 = 0$ is obtained in the limit $c \rightarrow + \infty$. Since we have proven analytically 
that $\tilde K - 2 = O(x_1^2)$ in that limit, the have extended the curves obtained with $c \in ]-1, 2[$ by straight lines for $c >2$.
Meanwhile the values $\underset{c \in ]-1,1[}\min \tilde K(c)$ were recorded for each $m$, and then used to draw the border between regions 
$(U)$ and $(B)$ in figure \ref{fig:bif_UBC}. 
In the same figure the borders between regions $(B_l)$ and $(B_p)$ and regions $(C_l)$ and $(C_p)$ were obtained for $\tilde K(c=1)$. 

\par Simulations of solutions of (\ref{eq:main}) shown in figures in the main text 
have been done using a finite difference scheme to discretize the diffusion operator 
$$ \partial_\theta^2 u^m(\theta)  \approx \frac{u^m(\theta + \Delta_\theta) + u^m(\theta - \Delta_\theta) - 2 u^m(\theta)}{(\Delta_\theta)^2},  $$
a symmetric scheme for the advection term
$$ \partial_\theta ( u J\ast u) \approx K x_1 \frac{ u(\theta + \Delta_\theta) \sin(\theta + \Delta_\theta)
 - u(\theta - \Delta_\theta) \sin(\theta - \Delta_\theta) }{ 2 \Delta_\theta} $$
and a explicit Euler scheme for the time derivative. 
Unless something else is explicitly mentioned, we have used the initial condition 
$u_0(\theta) = \frac{1}{2\pi} + \frac{0.2}{\pi}  \cos(\theta) $.
We mention that the scheme slightly unstabilize the uncoherent equilibria even when $\tilde K< \tilde K_c$. 
It also increases diffusion phenomena and induces some numerical instabilities at the border of the support of solutions
 when this support is strictly included in $]0, 2\pi[$.
Numerical results have been checked using $(\Delta_\theta, \Delta_t) = (2.0\, 10^{-2},  10^{-4} )$ and 
$(\Delta_\theta, \Delta_t) = (2.0\, 10^{-3},  10^{-6})$, 
convergence has been checked when $\Delta_\theta, \Delta_t \rightarrow 0$, and the convergence rate is 
$O(\Delta_\theta) + O(\Delta_t) + O( \frac{(\Delta_\theta)^2}{\Delta_t}) $.

\medskip


\bibliographystyle{amsplain}
\bibliography{Kuramoto_PM.bib}

\end{document}